\documentclass{easychair}

\usepackage[T1]{fontenc}
\usepackage[utf8]{inputenc}
\usepackage[main=english,russian]{babel}

\usepackage{doc}
\usepackage{amsmath}
\usepackage{amssymb}
\usepackage{enumitem}   
\usepackage{subcaption}

\usepackage{tikz}
\usetikzlibrary{calc,decorations.pathmorphing,decorations, arrows.meta, shapes, snakes} 

\title{Maslov's class K with Equivalence\thanks{Supported by Polish National Science Center, grant No. 2021/41/B/ST6/00996.}}
\author{Oskar Fiuk}
\institute{University of Wrocław\\ \email{contact.oskarfiuk@gmail.com}}
\authorrunning{Fiuk}
\titlerunning{Maslov's class K with Equivalence}

\begin{document}

\newcommand{\of}[1]{\mbox{}{\color{teal}\textbf{\small [OF: #1]}}}

\newtheorem{theorem}{Theorem}
\newtheorem{lemma}[theorem]{Lemma}
\newtheorem{corollary}[theorem]{Corollary}
\newtheorem{definition}[theorem]{Definition}
\newtheorem{proposition}[theorem]{Proposition}
\newtheorem{claim}[theorem]{Claim}
\newtheorem{assumption}[theorem]{Assumption}
\newtheorem{example}[theorem]{Example}
\newtheorem{remark}[theorem]{Remark}

\newcommand{\FO}{\mbox{\rm FO}}
\newcommand{\FOt}{\mbox{$\mbox{\rm FO}^2$}}
\newcommand{\GF}{\mbox{$\mbox{\rm GF}$}}
\newcommand{\FF}{\mbox{$\mbox{\rm FF}$}}
\newcommand{\ALC}{$\mathcal{ALC}$}
\newcommand{\UF}{\mbox{$\mbox{\rm UF}_1$}}

\newcommand{\ACK}{\mbox{$\mbox{\rm ACK}$}}
\newcommand{\GC}{\mbox{$\mbox{\rm GC}$}}
\newcommand{\SC}{\mbox{$\mbox{\rm SC}$}}
\newcommand{\MK}{\mbox{$\mbox{\rm K}$}}
\newcommand{\DMK}{\mbox{$\overline{\mbox{\rm K}}$}}
\newcommand{\DMDK}{\mbox{$\overline{\mbox{\rm DK}}$}}
\newcommand{\SNF}{\mbox{$\DMDK_{\rm snf}$}}
\newcommand{\MON}{\mbox{$\mbox{\rm MON}$}}
\newcommand{\AUFmin}{\mbox{$\mbox{\rm AUF}^-_1$}}

\newcommand{\LogSpace}{\textsc{LogSpace}}
\newcommand{\NLogSpace}{\textsc{NLogSpace}}
\newcommand{\NP}{\textsc{NPTime}}
\newcommand{\PTime}{\textsc{PTime}}
\newcommand{\PSpace}{\textsc{PSpace}}
\newcommand{\ExpTime}{\textsc{ExpTime}}
\newcommand{\ExpSpace}{\textsc{ExpSpace}}
\newcommand{\NExpTime}{\textsc{NExpTime}}
\newcommand{\TwoExpTime}{2\textsc{-ExpTime}}
\newcommand{\TwoNExpTime}{2\textsc{-NExpTime}}
\newcommand{\ThreeNExpTime}{3\textsc{-NExpTime}}
\newcommand{\APSpace}{\textsc{APSpace}}
\newcommand{\AExpSpace}{\textsc{AExpSpace}}
\newcommand{\TOWER}{\textsc{Tower}}
\newcommand{\RE}{\textsc{RE}}

\renewcommand{\phi}{\varphi} 
\newcommand{\eps}{\varepsilon} 
\renewcommand{\emptyset}{\varnothing} 

\newcommand{\AAA}{{\mbox{\large \boldmath $\alpha$}}}
\newcommand{\BBB}{{\mbox{\large \boldmath $\beta$}}}
\newcommand{\CCC}{{\mbox{\large \boldmath $\gamma$}}}
\newcommand{\DDD}{{\mbox{\large \boldmath $\delta$}}}
\newcommand{\TTT}{{\mbox{\large \boldmath $\tau$}}}
\newcommand{\TTTo}{{\mbox{$\TTT_{\ro}$}}}
\newcommand{\ro}{{\rm o}}

\newcommand{\TT}{{\mbox{\large \boldmath $T$}}}
\newcommand{\TTo}{{\mbox{$\TT_{\ro}$}}}

\newcommand{\restr}{\restriction}
\newcommand{\N}{{\mathbb N}}   
\newcommand{\Q}{{\mathbb Q}}   
\newcommand{\Z}{{\mathbb Z}}   
\newcommand{\R}{{\mathbb R}}   

\newcommand{\tp}[1]{{\rm tp}^{{#1}}}
\newcommand{\otp}[1]{{\rm tp}_{\rm o}^{{#1}}}
\newcommand{\etp}[1]{{\rm tp}_{=}^{{#1}}}
\newcommand{\lit}{{\rm lit}}

\newcommand{\Dom}{{\rm dom}}
\newcommand{\Vars}{{\rm{vars}}}
\newcommand{\Atoms}{{\rm{atoms}}}
\newcommand{\Arity}{{\rm{arity}}}
\newcommand{\chc}{{\rm{chc}}}

\newcommand{\FFF}{{F}}
\newcommand{\GGG}{{G}}
\newcommand{\EEE}{{E}}
\newcommand{\TR}{{T}}
\newcommand{\TRrefl}{{T^*}}
\newcommand{\CC}{{\operatorname{Con}}}
\newcommand{\LLL}{{L}}

\newcommand{\Cons}{{\bar{c}}}
\newcommand{\Rels}{{\bar{R}}}

\newcommand{\Qfr}{Q}

\newcommand{\tightdisplay}[1]{%
  \begingroup%
  \setlength{\abovedisplayskip}{0.4em}%
  \setlength{\belowdisplayskip}{0.4em}%
  #1\endgroup}

\renewcommand{\paragraph}[1]{%
\smallskip\noindent\textbf{#1}}

\newcommand{\varx}{\mathbf{x}}
\newcommand{\vary}{\mathbf{y}}

\newcommand{\cA}{{\mathcal{A}}}
\newcommand{\cB}{{\mathcal{B}}}
\newcommand{\cC}{{\mathcal{C}}}
\newcommand{\cD}{{\mathcal{D}}}
\newcommand{\cE}{{\mathcal{E}}}
\newcommand{\cF}{{\mathcal{F}}}
\newcommand{\cG}{{\mathcal{G}}}
\newcommand{\cH}{{\mathcal{H}}}
\newcommand{\cI}{{\mathcal{I}}}
\newcommand{\cJ}{{\mathcal{J}}}
\newcommand{\cK}{{\mathcal{K}}}
\newcommand{\cL}{{\mathcal{L}}}
\newcommand{\cM}{{\mathcal{M}}}
\newcommand{\cN}{{\mathcal{N}}}
\newcommand{\cO}{{\mathcal{O}}}
\newcommand{\cP}{{\mathcal{P}}}
\newcommand{\cQ}{{\mathcal{Q}}}
\newcommand{\cR}{{\mathcal{R}}}
\newcommand{\cS}{{\mathcal{S}}}
\newcommand{\cT}{{\mathcal{T}}}
\newcommand{\cU}{{\mathcal{U}}}
\newcommand{\cV}{{\mathcal{V}}}
\newcommand{\cW}{{\mathcal{W}}}
\newcommand{\cX}{{\mathcal{X}}}
\newcommand{\cY}{{\mathcal{Y}}}
\newcommand{\cZ}{{\mathcal{Z}}}

\newcommand{\bA}{{\mathbf{A}}}
\newcommand{\bB}{{\mathbf{B}}}
\newcommand{\bC}{{\mathbf{C}}}
\newcommand{\bD}{{\mathbf{D}}}
\newcommand{\bE}{{\mathbf{E}}}
\newcommand{\bF}{{\mathbf{F}}}
\newcommand{\bG}{{\mathbf{G}}}
\newcommand{\bH}{{\mathbf{H}}}
\newcommand{\bI}{{\mathbf{I}}}
\newcommand{\bJ}{{\mathbf{J}}}
\newcommand{\bK}{{\mathbf{K}}}
\newcommand{\bL}{{\mathbf{L}}}
\newcommand{\bM}{{\mathbf{M}}}
\newcommand{\bN}{{\mathbf{N}}}
\newcommand{\bO}{{\mathbf{O}}}
\newcommand{\bP}{{\mathbf{P}}}
\newcommand{\bQ}{{\mathbf{Q}}}
\newcommand{\bR}{{\mathbf{R}}}
\newcommand{\bS}{{\mathbf{S}}}
\newcommand{\bT}{{\mathbf{T}}}
\newcommand{\bU}{{\mathbf{U}}}
\newcommand{\bV}{{\mathbf{V}}}
\newcommand{\bW}{{\mathbf{W}}}
\newcommand{\bX}{{\mathbf{X}}}
\newcommand{\bY}{{\mathbf{Y}}}
\newcommand{\bZ}{{\mathbf{Z}}}

\newcommand{\bbA}{{\mathbb{A}}}
\newcommand{\bbB}{{\mathbb{B}}}
\newcommand{\bbC}{{\mathbb{C}}}
\newcommand{\bbD}{{\mathbb{D}}}
\newcommand{\bbE}{{\mathbb{E}}}
\newcommand{\bbF}{{\mathbb{F}}}
\newcommand{\bbG}{{\mathbb{G}}}
\newcommand{\bbH}{{\mathbb{H}}}
\newcommand{\bbI}{{\mathbb{I}}}
\newcommand{\bbJ}{{\mathbb{J}}}
\newcommand{\bbK}{{\mathbb{K}}}
\newcommand{\bbL}{{\mathbb{L}}}
\newcommand{\bbM}{{\mathbb{M}}}
\newcommand{\bbN}{{\mathbb{N}}}
\newcommand{\bbO}{{\mathbb{O}}}
\newcommand{\bbP}{{\mathbb{P}}}
\newcommand{\bbQ}{{\mathbb{Q}}}
\newcommand{\bbR}{{\mathbb{R}}}
\newcommand{\bbS}{{\mathbb{S}}}
\newcommand{\bbT}{{\mathbb{T}}}
\newcommand{\bbU}{{\mathbb{U}}}
\newcommand{\bbV}{{\mathbb{V}}}
\newcommand{\bbW}{{\mathbb{W}}}
\newcommand{\bbX}{{\mathbb{X}}}
\newcommand{\bbY}{{\mathbb{Y}}}
\newcommand{\bbZ}{{\mathbb{Z}}}

\newcommand{\fA}{{\mathfrak{A}}}
\newcommand{\fB}{{\mathfrak{B}}}
\newcommand{\fC}{{\mathfrak{C}}}
\newcommand{\fD}{{\mathfrak{D}}}
\newcommand{\fE}{{\mathfrak{E}}}
\newcommand{\fF}{{\mathfrak{F}}}
\newcommand{\fG}{{\mathfrak{G}}}
\newcommand{\fH}{{\mathfrak{H}}}
\newcommand{\fI}{{\mathfrak{I}}}
\newcommand{\fJ}{{\mathfrak{J}}}
\newcommand{\fK}{{\mathfrak{K}}}
\newcommand{\fL}{{\mathfrak{L}}}
\newcommand{\fM}{{\mathfrak{M}}}
\newcommand{\fN}{{\mathfrak{N}}}
\newcommand{\fO}{{\mathfrak{O}}}
\newcommand{\fP}{{\mathfrak{P}}}
\newcommand{\fQ}{{\mathfrak{Q}}}
\newcommand{\fR}{{\mathfrak{R}}}
\newcommand{\fS}{{\mathfrak{S}}}
\newcommand{\fT}{{\mathfrak{T}}}
\newcommand{\fU}{{\mathfrak{U}}}
\newcommand{\fV}{{\mathfrak{V}}}
\newcommand{\fW}{{\mathfrak{W}}}
\newcommand{\fX}{{\mathfrak{X}}}
\newcommand{\fY}{{\mathfrak{Y}}}
\newcommand{\fZ}{{\mathfrak{Z}}}

\newcommand{\as}{{\bar{a}}}
\newcommand{\bs}{{\bar{b}}}
\newcommand{\csbar}{{\bar{c}}}
\newcommand{\ds}{{\bar{d}}}
\newcommand{\es}{{\bar{e}}}
\newcommand{\fs}{{\bar{f}}}
\newcommand{\gs}{{\bar{g}}}
\newcommand{\hs}{{\bar{h}}}
\newcommand{\is}{{\bar{i}}}
\newcommand{\js}{{\bar{j}}}
\newcommand{\ks}{{\bar{k}}}
\newcommand{\ls}{{\bar{l}}}
\newcommand{\ms}{{\bar{m}}}
\newcommand{\ns}{{\bar{n}}}
\newcommand{\os}{{\bar{o}}}
\newcommand{\ps}{{\bar{p}}}
\newcommand{\qs}{{\bar{q}}}
\newcommand{\rs}{{\bar{r}}}
\newcommand{\ssbar}{{\bar{s}}} 
\newcommand{\ts}{{\bar{t}}}
\newcommand{\us}{{\bar{u}}}
\newcommand{\vs}{{\bar{v}}}
\newcommand{\ws}{{\bar{w}}}
\newcommand{\xs}{{\bar{x}}}
\newcommand{\ys}{{\bar{y}}}
\newcommand{\zs}{{\bar{z}}}

\maketitle

\begin{abstract}
  Maslov's class K is a decidable fragment of equality-free first-order logic.
  It subsumes numerous classical decidable fragments, including the equality-free monadic fragment, the equality-free two-variable fragment, and the Gödel class.
  Introduced over 50 years ago, the class K has been studied in automated deduction, primarily through resolution-based methods.
  Despite this long history, the precise complexity of the satisfiability problem for K remained open until recently, when NExpTime-completeness was finally established.

  The two main contributions of this work are as follows.
  \begin{itemize}[nosep]
    \item The recent proofs of NExpTime-completeness and the finite model property for K are highly non-trivial and technically sophisticated.
    We present considerably simpler proofs of both results.
    Our approach combines a novel reduction from K to the class of solvable Skolem sentences with a randomised construction of finite models.

    \item We introduce the class K+E by allowing sentences in K to contain one distinguished binary predicate E whose interpretation is constrained to be an equivalence relation.
    By leveraging our approach to this extended class, we prove that K+E has the finite model property and that its satisfiability problem is 2-NExpTime-complete.
  \end{itemize}
\end{abstract}

\section{Introduction}
\label{sec:introduction}

Hilbert's \emph{Entscheidungsproblem} asked for a general algorithm to determine the satisfiability of first-order statements.
Church and Turing famously showed that no such algorithm exists by establishing the undecidability of first-order logic.
This negative result gave rise to a broad research programme aimed at identifying fragments of first-order logic for which satisfiability remains decidable.
Over the decades, this effort has led to the discovery of numerous decidable formalisms, including the Bernays--Sch\"onfinkel class, the monadic fragment, the Ackermann class, the G\"odel class, the two-variable fragment, and the guarded fragment, among others.

One motivation for studying decidable fragments comes from automated deduction,
a branch of computer science that seeks effective methods for solving the satisfiability problem, not in full generality, but for broad and practically relevant classes of instances.
One of the pioneers of automated deduction was S.~J.~Maslov, who introduced the class \MK{}~\cite{Mas68,Mas71}.
In the 1960s, Maslov developed a proof procedure, called the \emph{inverse method}~\cite{Lifschitz89}, for deciding validity and applied it to the class \MK{}.
Subsequent work shifted attention to the dual class, denoted \DMK{}, and to its satisfiability problem.
In the 1980s, Zamov proposed a resolution procedure for \DMK{}; see Chapter~6 of the monograph \emph{Resolution Methods for the Decision Problem}~\cite{FLTZ93}.
In 1999, Hustadt and Schmidt~\cite{HuS99} continued Zamov's work and developed a refined variant of his procedure. 

The importance of Maslov's class \DMK{} stems from the fact that it is one of the broadest fragments of (equality-free) first-order logic.
In particular, \DMK{} subsumes many classical fragments (all considered without equality), including:
the monadic fragment~\cite{Low15}, the Ackermann class~\cite{Ack28} and its generalised variant due to Voigt~\cite{Voi19}, the G\"odel class~\cite{God33}, and the solvable Skolem class from Dreben and Goldfarb's monograph~\cite{DG79} (there denoted as Class~2.4).
Moreover, the even broader class of conjunctions of \DMK{} sentences, denoted \DMDK{}, is decidable~\cite{HuS99}.
The absence of equality in \DMK{} is unavoidable: already the G\"odel class with equality is undecidable~\cite{Gol84}.

The class \DMK{} also subsumes the equality-free two-variable fragment, \FOt{}, in the sense that sentences in Scott normal form belong to \DMK{}~\cite{Sco62}.
As a consequence, many standard modal and description logics can be translated into \DMK{}.
This connection is discussed in the survey~\cite{HSG04}, where \DMK{} is considered alongside other well-known fragments such as the guarded fragment and the fluted fragment, all of which are incomparable in expressive power with \DMK{}.

After the early 2000s, \DMK{} has received relatively little attention.
A renewed interest in \DMK{} is motivated by the study of the uniform one-dimensional fragment, denoted \UF{}. Introduced by Hella and Kuusisto in 2014~\cite{HK14}, see also~\cite{KK14,Kuu16},
\UF{} was designed as a canonical generalisation of \FOt{} to predicates of arity greater than $2$.
The satisfiability problem for \UF{} (with equality) is \NExpTime-complete~\cite{KK14}, matching the complexity of \FOt{}.
To increase the expressive power of (equality-free) \UF{}, Kiero\'nski~\cite{FK25} proposed a generalisation denoted \AUFmin{}.
He observed that \AUFmin{}, when put into a suitable normal form, falls within \DMDK{}.
In~\cite{FKM24}, an even broader extension of (equality-free) \UF{}, denoted $\forall\text{-}\UF$, is introduced. It, too, is subsumed by \DMDK{}.

Despite this broad applicability, the class \DMK{} remained poorly understood from a theoretical perspective, as the two central questions for logical fragments—the finite model property and the precise complexity of the satisfiability problem—remained open for over 50 years.\footnote{%
In particular, it appears that~\cite{FLTZ93,HuS99,FLHT01} refrain from making explicit claims regarding the complexity of \DMK{}, or even the running time of their decision procedures.
Nevertheless, elementary upper bounds can be derived from Theorem~6.3 in~\cite{FLTZ93} and Theorem~6 in~\cite{HuS99}; 
however, these bounds are at best doubly or triply exponential.}

These questions were resolved only recently: in 2024, Fiuk, Kiero\'nski, and Michielini~\cite{FKM24} established that \DMK{} has the finite model property and that the satisfiability problem for \DMK{} is \NExpTime{}-complete.
In light of this result, the significance of \DMK{} becomes even more apparent:
it can express many non-trivial properties, not expressible in other fragments,
while being no harder than \FOt{}, \UF{}, or the G\"odel class.
This contrasts with the guarded fragment, which is \TwoExpTime{}-complete~\cite{Gra99}, and with the fluted fragment, which is \TOWER{}-complete~\cite{P-HST19}.

\paragraph{Our contribution.}
The proof of the finite model property and the $\NExpTime{}$ upper bound for $\DMK{}$
given in~\cite{FKM24} is technically involved and difficult to follow;
we discuss it in more detail in Section~\ref{sec:complexity}.
Our first contribution is a new and substantially simpler proof of these results.

Our approach is based on a novel Skolemisation procedure that transforms
sentences in $\DMK{}$ into the solvable Skolem class~\cite{DG79,Gol93}.
Notably, the precise complexity of the solvable Skolem class was itself open prior
to~\cite{FKM24}.
We then apply Goldfarb's probabilistic construction of finite
models for sentences in the solvable Skolem class~\cite{Gol93}, which
generalises the earlier approach of Gurevich and Shelah for the G\"odel
class~\cite{GS83}.
A direct application of this randomised model construction, however,
does not yield the desired model-size bound.
To obtain a tight exponential bound, we refine the construction using
ideas from~\cite{FKM24}, which, thanks to our Skolemisation procedure,
can be applied in a considerably simpler setting.

Our second contribution extends the decidability result beyond the class $\DMK{}$.
A common way of increasing the expressive power of decidable fragments
is to restrict the class of admissible structures.
For instance, $\FOt{}$ has been studied over structures equipped with
linear orders~\cite{Ott01,SZ12}, equivalence
relations~\cite{KO12,KMPHT14,KT09}, and transitive
relations~\cite{GMV99,P-H18}.
We follow this line of research and prove that $\DMK{}$ remains decidable when extended with one equivalence relation.
More precisely, we introduce the class $\DMK{} + \EEE$, obtained by
adding to $\DMK{}$ one distinguished binary predicate $\EEE$ whose
interpretation is required to be an equivalence relation.
We show that $\DMK{} + \EEE$ has the finite model property
and that its satisfiability problem is $\TwoNExpTime{}$-complete.

\paragraph{Organisation.}
Section~\ref{sec:preliminaries} introduces the notation and basic
definitions.
Section~\ref{sec:syntax} presents the class $\DMK{}$ and related fragments.
Section~\ref{sec:skolemization} develops the Skolemisation procedure.
Section~\ref{sec:complexity} gives a simplified proof of the $\NExpTime{}$ upper bound for $\DMK{}$.
Section~\ref{sec:equivalence} establishes that $\DMK{} + \EEE$ is $\TwoNExpTime$-complete and shows that closely related extensions are already undecidable.
Finally, Section~\ref{sec:remarks} concludes with remarks on the possible addition of equality to \DMK{}.

\section{Preliminaries}
\label{sec:preliminaries}

\paragraph{Notation.}
Let $\N = \{ 0, 1, 2, \dots \}$.
For $k \in \N$, the notation $[k]$ stands for the set $\{1, \dots, k\}$, with the convention that $[0] = \emptyset$. More generally, we use interval notation $[a, b] \subseteq \N$ to denote the set $\{a, a+1, \dots, b\}$ for $a \leq b$, and the empty set $\emptyset$ for $a > b$.
The powerset of a set $S$ is denoted by $2^S$.
Tuples are denoted by bars and, by abuse of notation, sometimes regarded as sets of their elements.
Concatenation of tuples is denoted by juxtaposition.

\paragraph{First-order logic.}
We assume general familiarity with first-order logic ($\FO$).  
The logical symbols are $=$, $\bot$, $\top$, $\vee$, $\wedge$, $\neg$, $\rightarrow$, $\leftrightarrow$, and the quantifiers $\forall$, $\exists$.
The non-logical symbols are relation symbols (from a countably infinite set), constant symbols (also from a countably infinite set), and variables (again countably many).
We do not allow function symbols of positive arity.

A \emph{signature} is a pair $\Sigma = \langle \Rels,\Cons \rangle$,
where $\Rels$ is a non-empty tuple of relation symbols and $\Cons$ is a tuple of constant symbols.
When convenient, we identify $\Sigma$ with the set $\Rels \cup \Cons$.
Every relation symbol $R \in \Rels$ comes with associated arity, denoted $\Arity(R)$.

The \emph{length} of a formula $\phi$, denoted $|\phi|$, is the total number of symbols it contains, where each occurrence of a symbol---be it a variable, relation symbol, or constant---contributes~$1$.
The signature of $\phi$ consists of relation and constant symbols that appear in the formula.
The set of atomic subformulas of $\phi$ is denoted by $\Atoms(\phi)$.
The set of \emph{free} variables of $\phi$ is denoted by $\Vars(\phi)$;
we often view this set as a tuple $\xs$ and indicate this by $\phi[\xs]$.
Given a tuple of elements $\as$ of the same length as $\xs$, we write
$\phi[\as]$ for the result of componentwise substitution of the
variables in $\xs$ by the corresponding elements of $\as$.
More generally, for any tuple of variables $\ys$, not necessarily
coinciding with the free variables of $\phi$, and any tuple of elements
$\bs$ of the same length as $\ys$, we write $\phi[\bs/\ys]$ for the
result of componentwise substitution of the variables in $\ys$ by the
corresponding elements of $\bs$.
We write $\pm \phi$ to denote either $\phi$ or $\neg \phi$.

We use Fraktur letters such as $\fA, \fB, \dots$ to denote structures, and the corresponding Roman letters $A, B, \dots$ for their domains.
Let $\Sigma = \langle\Rels,\Cons\rangle$ be a signature.
A \emph{$\Sigma$-structure} $\fA$ interprets every relation symbol $R \in \Rels$ as a relation $R^{\fA}\subseteq A^{\Arity(R)}$ and every constant symbol $c \in \Cons$ as an element $c^{\fA} \in A$.
An element $a \in A$ is \emph{unnamed} if $a \neq c^\fA$ for every constant $c \in \Cons$.

For a sub-signature $\Sigma_0 \subseteq \Sigma$,
the \emph{$\Sigma_0$-reduct} of $\fA$ is a $\Sigma_0$-structure $\fA_0$ obtained from $\fA$ by
forgetting the interpretations of the symbols in $\Sigma \setminus \Sigma_0$.
Equivalently, $\fA$ is a \emph{$\Sigma$-expansion} of $\fA_0$.

\paragraph{Finite model property.}
Let $\cF$ be a fragment (or class) of $\FO{}$ sentences.
The fragment $\cF$ has the \emph{finite model property} if every satisfiable sentence $\phi \in \cF$ admits a finite model.
If there exists a fixed polynomial $p$ such that every satisfiable sentence $\phi \in \cF$ admits a model over a domain of size at most $2^{p(|\phi|)}$, then $\cF$ has the \emph{exponential model property}.
If, instead, the domain size is bounded by $2^{2^{p(|\phi|)}}$, then $\cF$ has the \emph{doubly exponential model property}.

The \emph{satisfiability problem} for $\cF$ is to decide whether an input sentence $\phi \in \cF$ admits a model.
The finite model property for $\cF$ implies that the satisfiability problem for $\cF$ is decidable.
The exponential model property for $\cF$ implies that the satisfiability problem for $\cF$ is in $\NExpTime{}$, while the doubly exponential model property implies that it is in $2$-$\NExpTime{}$.

\paragraph{Types and outer-types.}
Let $\Sigma=\langle\Rels,\Cons\rangle$ be a signature.
For $k\in\N$, let $\lit^k(\Sigma)$ denote the set of all literals over
$\Sigma$ and the variables $x_1,\dots,x_k$; literals involving $=$ or
$\neq$ are not included.

For $k\in\N$, a \emph{$k$-type} over $\Sigma$ is a maximal consistent
subset of $\lit^k(\Sigma)$; that is, for every pair of dual literals
$\lambda,\neg\lambda\in\lit^k(\Sigma)$, exactly one of $\lambda$ and
$\neg\lambda$ belongs to the $k$-type.
A $0$-type contains literals involving only constants.
A $(k+1)$-type $\sigma$ \emph{extends} a $k$-type $\tau$ if
$\tau\subseteq\sigma$.
A \emph{type} is simply a $k$-type for some $k\ge0$.
When convenient, we identify types with the conjunctions of the
literals they contain.
Let $\TTT^k(\Sigma)$ denote the set of all $k$-types over $\Sigma$, and
write
\(
  \TTT(\Sigma)
  \coloneq
  \bigcup_{k\in\N}\TTT^k(\Sigma).
\)

Let $\fA$ be a $\Sigma$-structure, and let
$\as=\langle a_1,\dots,a_k\rangle$ be a $k$-tuple of pairwise distinct
unnamed elements of $\fA$.
We write
\(
  \tp{\fA}[a_1,\dots,a_k]
  \coloneq
  \bigl\{
    \lambda\in\lit^k(\Sigma)
    \mid
    \fA\models\lambda[\as/\xs]
  \bigr\}
\)
for the $k$-type \emph{realised} by $\as$ in $\fA$, where
$[\as/\xs]$ denotes the substitution $x_i\mapsto a_i$.

Our notion of types is a slight adaptation of the standard notion of
atomic types, tailored to the setting needed in this paper.
First, our types do not contain equality literals.
Second, we define the type realised by a tuple only when its elements
are pairwise distinct and unnamed.

In our proofs, we also need a refined variant of types.
For $k\in\N$, let
\tightdisplay{\[
  \lit^k_\ro(\Sigma)
  \coloneq
  \bigl\{
    \lambda\in\lit^k(\Sigma)
    \mid
    \Vars(\lambda)=\{x_1,\dots,x_k\}
    \text{ or }
    |\Vars(\lambda)|\le1
  \bigr\}.
\]}%
A \emph{$k$-outer-type} over $\Sigma$ is a maximal consistent subset of
$\lit^k_\ro(\Sigma)$.
Thus, a $k$-outer-type is akin to a $k$-type, but contains only literals
involving either all variables $x_1,\dots,x_k$ or at most one of them.
In particular, $k$-types and $k$-outer-types coincide for $k\le2$.

We use for outer-types the terminology and conventions introduced above
for types.
In particular, an \emph{outer-type} is a $k$-outer-type for some
$k\ge0$, and we identify outer-types with the conjunctions of their
literals when convenient.
We write $\TTT_\ro^k(\Sigma)$ and $\TTT_\ro(\Sigma)$ for the sets of
all $k$-outer-types and all outer-types over $\Sigma$, respectively.

For a $k$-tuple $\as=\langle a_1,\dots,a_k\rangle$ of pairwise distinct unnamed elements of $\fA$,
the $k$-outer-type \emph{realised} by $\as$ in $\fA$ is
\(
  \otp{\fA}[a_1,\dots,a_k]
  \coloneq
  \tp{\fA}[a_1,\dots,a_k]\cap\lit^k_\ro(\Sigma).
\)

For a $k$-outer-type $\tau_\ro$, we write $\tau_\ro\langle x_i \rangle$ for the $1$-type \emph{induced on the variable $x_i$}.
It is obtained by taking all literals $\lambda \in \tau_\ro$ such that $\Vars(\lambda) \subseteq \{ x_i \}$ and renaming $x_i$ with $x_1$.
Intuitively, if a tuple $\as$ realises $\tau_\ro$ in a structure, then $\tau_\ro\langle x_i \rangle$ is the $1$-type of the $i$th element of $\as$.

For a $k$-type $\tau$ and pairwise distinct indices $i_1,\dots,i_\ell \in [k]$, we write $\tau \langle x_{i_1},\dots,x_{i_\ell} \rangle$ for the $\ell$-outer-type \emph{induced by $\tau$ on $\langle x_{i_1},\dots,x_{i_\ell}\rangle$}.
It is obtained by taking all literals $\lambda \in \tau$ such that either
$\Vars(\lambda) = \{x_{i_1},\dots,x_{i_\ell}\}$ or
$\Vars(\lambda) \subseteq \{x_{i_j}\}$ for some $j \in [\ell]$,
and then renaming the variables according to the mapping
$x_{i_j} \mapsto x_j$ for every $j \in [\ell]$.
Intuitively, if a tuple
$\langle a_1,\dots,a_k\rangle$ realises $\tau$ in a structure, then
$\tau \langle x_{i_1},\dots,x_{i_\ell} \rangle$ is the $\ell$-outer-type of the subtuple
$\langle a_{i_1},\dots,a_{i_\ell}\rangle$.

For the set of all outer-types induced by the $k$-type $\tau$, we write
\tightdisplay{\[ \TTTo(\tau) \;\coloneq\; \big\{ \tau\langle x_{i_1},\dots,x_{i_\ell} \rangle \mid \ell \in [0,k] \text{ and } i_1,\dots,i_\ell \in [k] \text{ are pairwise distinct} \big\}.\]}%

\section{The Class \DMK{} and Related Fragments}
\label{sec:syntax}

\paragraph{Definition.}
Let $\phi$ be an equality-free sentence in negation normal form.
The sentence $\phi$ may contain constant symbols.
Without loss of generality, 
we assume that distinct quantifier
occurrences in $\phi$ bind distinct variables.

For an atom $\alpha\in\Atoms(\phi)$, the \emph{$\phi$-prefix} of
$\alpha$ is the sequence of quantifiers in $\phi$ that bind the
variables in $\Vars(\alpha)$.
The order of quantifiers in the $\phi$-prefix matches the order of their
occurrences in $\phi$ from left to right.
In contrast, the order and multiplicity of variables within $\alpha$,
as well as the constants occurring in $\alpha$, are irrelevant for
determining its $\phi$-prefix.
In particular, if $\alpha$ contains only constant symbols, then its
$\phi$-prefix is the empty sequence.

The sentence $\phi$ belongs to the class $\DMK{}$ if there exists a selection
\( \forall x_1,\dots,\forall x_n \)
of universal quantifier occurrences in $\phi$, for some $n \geq 0$, such that:
\begin{enumerate}[nosep]
    \item None of the selected universal quantifiers is in the scope of an existential quantifier.
    \item For every atom $\alpha \in \Atoms(\phi)$, at least one of the following holds:
    \begin{enumerate}[label=(\roman*),nosep]
      \item\label{DMK-category1} the $\phi$-prefix of $\alpha$ has length at most $1$;
      \item\label{DMK-category2} the $\phi$-prefix of $\alpha$ consists exactly of the selected
            universal quantifiers
            $\forall x_1,\dots,\forall x_n$;
      \item\label{DMK-category3} the last quantifier in the $\phi$-prefix of $\alpha$ is
            existential.
    \end{enumerate}
\end{enumerate}
The selected universal quantifiers are listed in their left-to-right
order in $\phi$, but need not form a contiguous block in the sequence
of all quantifiers of $\phi$.
The variables $x_1,\dots,x_n$ bound by the selected quantifiers are called the \emph{uniform variables}
of $\phi$, and $n$ is called the \emph{uniformity} of~$\phi$.

To avoid ambiguity, we adopt the convention that $n$ is always chosen
as small as possible.
Since for $n=1$ condition~\ref{DMK-category2} is subsumed by
condition~\ref{DMK-category1}, we have $n=0$ or $n\geq2$.
Under this convention, the required selection
$\forall x_1,\dots,\forall x_n$, whenever it exists, is unique.
Moreover, if $n\geq2$, it must coincide
with the $\phi$-prefix of some atom of $\phi$.

The above definition of $\DMK{}$ does not essentially rely on negation
normal form; it only requires that every quantifier occur positively.
We therefore adopt a more liberal syntax for $\DMK{}$ and allow
arbitrary Boolean combinations within quantifier-free subformulas.
In particular, we admit $\rightarrow$ and $\leftrightarrow$, provided
that no quantifier occurs in the premise of an implication or on
either side of a bi-implication.

\paragraph{Prenex form.}
It is often convenient to consider sentences of $\DMK{}$ in prenex
form.
Given a sentence $\phi\in\DMK{}$, its standard conversion to prenex
form yields, up to a reordering of the initial universal quantifiers, a
sentence of the form
\tightdisplay{\begin{equation}\label{prenex-DMK}
  \forall x_1\; \dots\; \forall x_n\;
  \Qfr_{n+1}x_{n+1}\; \dots\; \Qfr_mx_m\; \psi,
\end{equation}}%
where $0\le n\le m$, each $\Qfr_i\in\{\forall,\exists\}$, and $\psi$
is equality-free and quantifier-free.

The conditions on $\phi$-prefixes translate as follows: for every
atom $\alpha\in\Atoms(\psi)$, at least one of the following holds:
\begin{enumerate}[label=(\roman*),nosep]
  \item[\ref{DMK-category1}]\label{prenex-DMK:item1}
    $|\Vars(\alpha)|\le1$;
  \item[\ref{DMK-category2}]\label{prenex-DMK:item2}
    $\Vars(\alpha)=\{x_1,\dots,x_n\}$, where
    $x_1,\dots,x_n$ are the uniform variables of $\phi$;
  \item[\ref{DMK-category3}]\label{prenex-DMK:item3}
    the variable $x_i\in\Vars(\alpha)$ with maximal index $i$ is
    existentially quantified.
\end{enumerate}

Conversely, it is straightforward to verify that a sentence in negation normal form belongs to
$\DMK{}$ whenever its standard conversion to prenex form satisfies the
conditions above.

\paragraph{Examples.}
We illustrate the definition of $\DMK{}$ with a few examples.

\smallskip
\emph{Example 1.}
Consider the sentence $\phi_1$:
\tightdisplay{\[
\begin{aligned}
\forall a,j,s\;
\bigl(&
  \bigl(
    \mathsf{article}(a)
    \wedge \mathsf{journal}(j)
    \wedge \mathsf{scientist}(s)
    \wedge \mathsf{submittedToBy}(a,j,s)
  \bigr)
  \rightarrow \\[-1mm]
&\exists r\;
  \bigl(
    \mathsf{reviewer}(r)
    \wedge \mathsf{reviewsFor}(r,j)
    \wedge \neg\mathsf{hasConflict}(r,s)
    \wedge \mathsf{qualifiedFor}(r,a)
  \bigr)
\bigr).
\end{aligned}
\]}%

The sentence $\phi_1$ belongs to $\DMK{}$ with uniform variables
$a,j,s$.
The $\phi_1$-prefixes of its atoms satisfy the conditions of the definition as follows:
\begin{enumerate}[label=(\roman*),nosep]
  \item[\ref{DMK-category1}]
    the prefixes $\forall a$, $\forall j$, $\forall s$, and
    $\exists r$ have length at most~$1$;
  \item[\ref{DMK-category2}]
    the prefix $\forall a\forall j\forall s$ consists of all
    selected universal quantifiers;
  \item[\ref{DMK-category3}]
    the remaining prefixes $\forall j\exists r$,
    $\forall s\exists r$, and $\forall a\exists r$ end with an
    existential quantifier.
\end{enumerate}

\smallskip
\emph{Example 2.}
Consider the sentence $\phi_2$:
\tightdisplay{\[
\begin{aligned}
\forall s\; \bigl(&
  \mathsf{student}(s) \rightarrow
  \exists p\; \bigl(
    \mathsf{professor}(p)
    \wedge \mathsf{teaches}(p,s)
    \wedge \forall w\; \bigl(
      \mathsf{week}(w) \rightarrow \\[-1mm]
      &\exists h\; \bigl(
        \mathsf{homework}(h)
        \wedge \mathsf{assignedToBy}(h,s,p)
        \wedge \mathsf{deadlineFor}(w,h)
      \bigr)
    \bigr)
  \bigr)
\bigr).
\end{aligned}
\]}%

The sentence $\phi_2$ belongs to $\DMK{}$ with no uniform variables.
The $\phi_2$-prefixes of its atoms satisfy the conditions of the definition as follows:
\begin{enumerate}[label=(\roman*),nosep]
  \item[\ref{DMK-category1}]
    the prefixes $\forall s$, $\exists p$, $\forall w$, and
    $\exists h$ have length at most~$1$;
  \item[\ref{DMK-category3}]
    the remaining prefixes $\forall s\exists p$,
    $\forall s\exists p\exists h$, and
    $\forall w\exists h$ end with an existential quantifier.
\end{enumerate}

\smallskip
\emph{Example 3.}
Observe that $\DMK{}$ allows permutations of
variables within atoms:
\tightdisplay{\[
  \phi_3
  \;\coloneq\;
  \forall x_1,x_2,x_3,x_4\;
  \bigl(
    P(x_1,x_2,x_3,x_4)
    \rightarrow
    \bigl(
      P(x_2,x_1,x_3,x_4)
      \wedge P(x_2,x_3,x_4,x_1)
    \bigr)
  \bigr).
\]}%
Indeed, $\phi_3$ belongs to $\DMK{}$ with uniform variables
$x_1,x_2,x_3,x_4$: every atom has the same $\phi_3$-prefix
$\forall x_1\forall x_2\forall x_3\forall x_4$, which satisfies
condition~\ref{DMK-category2}.

\smallskip
\emph{Nonexample.}
Consider the transitivity axiom
\tightdisplay{\[
  \phi_4
  \;\coloneq\;
  \forall x,y,z\;
  \bigl(
    (T(x,y)\wedge T(y,z))
    \rightarrow T(x,z)
  \bigr).
\]}%
Its atoms have the $\phi_4$-prefixes
$\forall x\forall y$, $\forall y\forall z$, and
$\forall x\forall z$.
They have length greater than~$1$, do not end with an existential
quantifier, and cannot all satisfy condition~\ref{DMK-category2}.
Hence, $\phi_4\notin\DMK{}$.

\paragraph{The class $\DMDK{}$.}
Following~\cite{HuS99}, let $\DMDK{}$ denote the class of conjunctions
of $\DMK{}$ sentences; that is,
\tightdisplay{\[
  \phi_1\wedge\dots\wedge\phi_k\in\DMDK{}
  \qquad\text{whenever}\qquad
  \phi_1,\dots,\phi_k\in\DMK{}.
\]}%

A sentence of $\DMDK{}$ need not be equivalent to a sentence of
$\DMK{}$: conjuncts of different uniformities cannot, in general, share
a common sequence of uniform variables.
If they have the same uniformity, however, they can be combined into a
single $\DMK{}$ sentence.
To see this, consider $\phi_1,\phi_2\in\DMK{}$ of the same uniformity.
Without loss of generality, both are in prenex form~\eqref{prenex-DMK},
with the same sequence $\xs$ of uniform variables.
Writing
\(
  \phi_i=\forall\xs\;\phi'_i[\xs],
\)
for $i \in \{ 1, 2\}$,
we have
\tightdisplay{\[
  \phi_1\wedge\phi_2
  \quad\Longleftrightarrow\quad
  \forall\xs\;
  \bigl(\phi'_1[\xs]\wedge\phi'_2[\xs]\bigr).
\]}%
A routine check shows that the sentence on the right belongs to
$\DMK{}$.

\paragraph{The solvable Skolem class.}
In their monograph~\cite{DG79}, Dreben and Goldfarb systematically
study solvable classes of formulas, including several classes of
Skolem sentences ($\forall^*\exists^*$).
One of these, referred to as Class~2.4 and later also studied
in~\cite{Gol93}, is the only such class relevant to this paper.
We therefore refer to it simply as the \emph{solvable Skolem class}.

Sentences in the solvable Skolem class have the form
\tightdisplay{\begin{equation}
  \forall x_1\; \dots\; \forall x_n\;
  \forall x_{n+1}\; \dots\; \forall x_\ell\;
  \exists x_{\ell+1}\; \dots\; \exists x_m\; \psi,
\end{equation}}%
where $0\le n\le\ell\le m$, $\psi$ is equality-free and
quantifier-free, and every $\alpha\in\Atoms(\psi)$ satisfies at least
one of the following:
\begin{enumerate}[label=(\roman*),nosep]
  \item $|\Vars(\alpha)|\leq 1$;
  \item $\Vars(\alpha)=\{x_1,\dots,x_n\}$;
  \item $x_i\in\Vars(\alpha)$ for some $i>\ell$.
\end{enumerate}

The solvable Skolem class is a subfragment of $\DMK{}$: its sentences
are in prenex form~\eqref{prenex-DMK}, with $x_1,\dots,x_n$ as uniform
variables, subject to the additional restriction that all existential
quantifiers follow all universal quantifiers.
In contrast, $\DMK{}$ allows arbitrary quantifier alternation after
the initial block binding the uniform variables.

The solvable Skolem class will play an important role in our
decidability proof, as we show that every sentence of $\DMK{}$ can be
reduced to this class while preserving satisfiability.

\paragraph{Other fragments.}
From the perspective of satisfiability, Maslov's class $\DMK{}$ subsumes several
important fragments of equality-free $\FO{}$.
These include the Ackermann class (the prenex class
$\exists^*\forall\exists^*$), the Scott class (the normal form of
$\FOt{}$), and the G\"odel class (the prenex class
$\exists^*\forall\forall\exists^*$), where the Ackermann and Scott
classes are understood without equality.
The initial block of existential quantifiers in the Ackermann and
G\"odel classes can be replaced by fresh constants while preserving
satisfiability.
These fragments are discussed in detail in the
monographs~\cite{BGG,PH23}.

The class $\DMK{}$ also subsumes the normal form of the equality-free
uniform one-dimensional fragment~\cite{HK14,KK14}, as well as its
variants with alternating quantifier blocks studied
in~\cite{FK25,FKM24}.
In contrast, $\DMK{}$ is incomparable with several other well-known
decidable fragments.
These include the guarded and clique-guarded
fragments~\cite{ABN98,BGO14}, the unary negation
fragment~\cite{StC13}, and the fluted and adjacent
fragments~\cite{P-HST19,BKP23}.

\section{A Skolemisation Procedure}
\label{sec:skolemization}

In this section, we introduce the \emph{Skolem normal form} for sentences in $\DMDK{},$ which allows for a substantially simpler treatment of $\DMDK{}$ in subsequent proofs.
We then show a polynomial-time reduction from the general case to Skolem normal form that preserves (finite) satisfiability.

\begin{definition}\label{def:skolem-class}
  We say that a sentence $\phi\in\DMDK{}$ is in \emph{Skolem normal form} if
  \tightdisplay{\begin{equation*}
    \phi
    \;=\;
    \bigwedge\nolimits_{t\in\cI}
    \forall\xs\; \exists x_{m+1}\; \psi_t[\xs,x_{m+1}],
  \end{equation*}}%
  where
  $\cI$ is a finite set of indices,
  $\xs=\langle x_1,\dots,x_m\rangle$ is a variable sequence of length
  $m\ge1$, common to all $t$,
  $x_{m+1}$ is a single variable, and
  $\psi_t[\xs,x_{m+1}]$ is quantifier-free.

  For every $t\in\cI$, let
  \tightdisplay{\[
    \phi_t
    \;\coloneq\;
    \forall\xs\; \exists x_{m+1}\; \psi_t
  \]}%
  and let $n_t\in\{0\}\cup[2,m]$ denote the uniformity of $\phi_t$,
  with uniform variables $x_1,\dots,x_{n_t}$.

  We partition the set of atoms of $\psi_t$ as
  $\Atoms(\psi_t)=\AAA_t\cup\BBB_t\cup\CCC_t$, where
  \tightdisplay{\[
  \begin{aligned}
    \AAA_t
    &\coloneq
    \bigl\{
      \alpha\in\Atoms(\psi_t)
      \mid
      \Vars(\alpha)\subseteq\{x_i\}
      \text{ for some $i\in[m]$}
    \bigr\},
    \\
    \BBB_t
    &\coloneq
    \bigl\{
      \alpha\in\Atoms(\psi_t)
      \mid
      n_t \ge 2
      \text{ and }
      \Vars(\alpha)=\{x_1,\dots,x_{n_t}\}
    \bigr\},
    \\
    \CCC_t
    &\coloneq
    \bigl\{
      \alpha\in\Atoms(\psi_t)
      \mid
      x_{m+1}\in\Vars(\alpha)
    \bigr\}.
  \end{aligned}
  \]}%
\end{definition}

Let $\SNF{}$ denote the class of all Skolem normal form sentences of
$\DMDK{}$.
Whenever we consider a sentence of $\SNF{}$, we use the notation of
Definition~\ref{def:skolem-class}.

We now establish the required reduction from $\DMDK{}$ to $\SNF{}$.
Since the reduction can be viewed as a form of
\emph{relational Skolemisation}, we refer to it as the
\emph{Skolemisation procedure}.

The usual Skolem form can subsequently be obtained by replacing the
existential variable $x_{m+1}$ in each conjunct $\phi_t$ by a fresh term
$f_t(x_1,\dots,x_m)$.
We will not need this form.

\begin{proposition}\label{prop:reduction}
  There exists a polynomial-time function
  $\cF \colon \DMDK{} \rightarrow \SNF{}$ such that, for every non-empty domain $D$,
  $\phi \in \DMDK{}$ has a model over $D$ if and only if $\cF(\phi)$ has a model over $D$.
\end{proposition}

\begin{proof}
We show the rewriting for a single $\DMK{}$ sentence; the extension to
$\DMDK{}$ is straightforward.

Let $\phi\in\DMK{}$ and let $\Sigma=\langle\Rels,\Cons\rangle$ be the
signature of $\phi$.
By converting $\phi$ to the prenex form~\eqref{prenex-DMK} and
introducing dummy quantifiers whenever necessary, we may assume that
\tightdisplay{\[
  \phi
  \;=\;
  \forall\xs\;
  \forall x_{n+1}\;
  \exists x_{n+2}\;
  \cdots\;
  \forall x_{m-1}\;
  \exists x_m\;
  \psi,
\]}%
where $\xs=\langle x_1,\dots,x_n\rangle$ are the uniform variables,
$m-n>0$ is even, the quantifier binding $x_{n+i}$ is universal for
odd $i$ and existential for even $i$, and $\psi$ is quantifier-free.

We proceed by induction on $m-n$.
If $m-n=2$, then
\(
  \phi
  =
  \forall x_1\cdots\forall x_{n+1}\;
  \exists x_{n+2}\;\psi
\)
is already in $\SNF{}$, and we set
\(
  \cF(\phi)\coloneq\phi.
\)
Suppose now that $m-n>2$.
Let
\(
  \xs'=\langle x_1,\dots,x_{n+2}\rangle
\)
and
\tightdisplay{\[
  \chi[\xs']
  \;\coloneq\;
  \forall x_{n+3}\;
  \exists x_{n+4}\;
  \cdots\;
  \forall x_{m-1}\;
  \exists x_m\;
  \psi.
\]}%
Thus,
\(
  \phi
  =
  \forall\xs\;
  \forall x_{n+1}\;
  \exists x_{n+2}\;
  \chi.
\)

Define
\(
  \DDD
  \coloneq
  \bigl\{
    \alpha\in\Atoms(\psi)
    \mid
    \Vars(\alpha)\subseteq\{x_1,\dots,x_{n+2}\}
  \bigr\}.
\)
We extend $\Sigma$ to a signature $\Sigma^+$ by introducing a fresh
$(n+2)$-ary relation symbol $R_\chi$ and, for every
$\alpha\in\DDD$, a fresh $(n+2)$-ary relation symbol $R_\alpha$.
Let $\chi^+$ be obtained from $\chi$ by replacing every occurrence of
each atom $\alpha\in\DDD$ by $R_\alpha(\xs')$.
The new atoms use precisely the variables
$\xs'=\langle x_1,\dots,x_{n+2}\rangle$, in this order, while the
subscript $\alpha$ records the original relation symbol and its full
sequence of arguments, including repetitions and constants.

Over $\Sigma^+$, define
\tightdisplay{\[
  \mu
  \;\coloneq\;
  \forall\xs\;
  \forall x_{n+1}\;
  \exists x_{n+2}\;
  \Bigl(
    R_\chi(\xs')
    \wedge
    \bigwedge\nolimits_{\alpha\in\DDD}
    \bigl(
      \alpha\leftrightarrow R_\alpha(\xs')
    \bigr)
  \Bigr)
  \quad\text{and}\quad
  \nu
  \;\coloneq\;
  \forall\xs'\;
  \bigl(
    R_\chi(\xs')
    \rightarrow
    \chi^+
  \bigr).
\]}%

The above sentences provide an intermediate rewriting step in the
construction of $\cF(\phi)$.
We return to $\cF(\phi)$ shortly and first verify
the correctness of this step.

We first show that
\(
  \mu\wedge\nu\models\phi.
\)
By the definition of $\nu$, we may replace the conjunct
$R_\chi(\xs')$ in the matrix of $\mu$ by $\chi^+$ when deriving a
consequence of $\mu\wedge\nu$.
It follows that
\tightdisplay{\[
  \mu\wedge\nu
  \;\models\;
  \forall\xs\;
  \forall x_{n+1}\;
  \exists x_{n+2}\;
  \Bigl(
    \chi^+
    \wedge
    \bigwedge\nolimits_{\alpha\in\DDD}
    \bigl(
      \alpha\leftrightarrow R_\alpha(\xs')
    \bigr)
  \Bigr).
\]}%

For every $\alpha\in\DDD$, the equivalence
\(
  \alpha\leftrightarrow R_\alpha(\xs')
\)
allows each occurrence of $R_\alpha(\xs')$ in $\chi^+$ to be replaced
by the corresponding atom $\alpha$.
By the definition of $\chi^+$, performing these replacements reconstructs
$\chi$.
It follows that
\tightdisplay{\[
  \mu\wedge\nu
  \;\models\;
  \forall\xs\;
  \forall x_{n+1}\;
  \exists x_{n+2}\;
  \chi.
\]}%
The sentence on the right is precisely
\(
  \phi.
\)
Hence, $\mu\wedge\nu \models \phi$.

Conversely, suppose that $\fA\models\phi$.
We expand $\fA$ to a $\Sigma^+$-structure by interpreting
\tightdisplay{\[
  R_\chi^{\fA}
  \;\coloneq\;
  \bigl\{
    \as\in A^{n+2}
    \mid
    \fA\models\chi[\as]
  \bigr\},
  \qquad
  R_\alpha^{\fA}
  \;\coloneq\;
  \bigl\{
    \as\in A^{n+2}
    \mid
    \fA\models\alpha[\as/\xs']
  \bigr\}
  \quad\text{for $\alpha\in\DDD$.}
\]}%
By construction, the resulting expansion satisfies
\tightdisplay{\[
  \eta
  \;\coloneq\;
  \forall\xs'\;
  \bigl(
    R_\chi(\xs')\leftrightarrow\chi
  \bigr)
  \qquad\text{and}\qquad
  \theta
  \;\coloneq\;
  \forall\xs'\;
  \bigwedge\nolimits_{\alpha\in\DDD}
  \bigl(
    \alpha\leftrightarrow R_\alpha(\xs')
  \bigr).
\]}%
Thus, every model of $\phi$ admits a $\Sigma^+$-expansion satisfying
$\phi\wedge\eta\wedge\theta$.
It remains to show that
\(
  \phi\wedge\eta\wedge\theta\models\mu\wedge\nu.
\)

Since
\(
  \phi
  =
  \forall\xs\;
  \forall x_{n+1}\;
  \exists x_{n+2}\;
  \chi,
\)
we have
\tightdisplay{\[
  \phi\wedge\eta
  \models
  \forall\xs\;
  \forall x_{n+1}\;
  \exists x_{n+2}\;
  R_\chi(\xs').
\]}%
Together with $\theta$, this yields
\(
  \phi\wedge\eta\wedge\theta\models\mu.
\)

Moreover, $\theta$ allows every occurrence of an atom
$\alpha\in\DDD$ in $\chi$ to be replaced by $R_\alpha(\xs')$.
Hence, by the definition of $\chi^+$,
\tightdisplay{\[
  \theta
  \models
  \forall\xs'\;
  \bigl(\chi\leftrightarrow\chi^+\bigr).
\]}%
Combining this with $\eta$ yields
\(
  \eta\wedge\theta\models\nu.
\)
Therefore,
\(
  \phi\wedge\eta\wedge\theta\models\mu\wedge\nu.
\)
This establishes the correctness of the rewriting step.
We now return to the construction of $\cF(\phi)$.

The sentence $\mu$ belongs to $\SNF{}$.
The sentence $\nu$ need not belong to $\SNF{}$, but it is again a
$\DMK{}$ sentence, now with uniform variables $\xs'$.
Moreover, $\nu$ has strictly smaller quantifier alternation depth than $\phi$,
so the induction hypothesis applies to it.
We set
\(
  \cF(\phi)
  \coloneq
  \mu\wedge\cF(\nu).
\)

Finally, we introduce dummy universal quantifiers as necessary to ensure that all conjuncts have the same number of them, as required by Definition~\ref{def:skolem-class}.
This establishes the reduction.
\end{proof}

\paragraph{A simple decidability argument.}
Observe that sentences in $\SNF{}$ are, in fact, conjunctions of sentences in the solvable Skolem class.
Combining this observation with Proposition~\ref{prop:reduction} yields a simple decidability argument for $\DMDK{}$, based on the earlier work of Goldfarb~\cite{Gol93}.

Using a probabilistic method, Goldfarb establishes the finite model property and decidability of the solvable Skolem class.
Although his result is technically stated for a single sentence in the solvable Skolem class, whereas $\SNF{}$ consists of conjunctions of such sentences, his argument extends seamlessly to this broader setting.
Thus, we obtain the following result.

\begin{corollary}
  The class \DMDK{} has the finite model property. The satisfiability problem for the class \DMDK{} is decidable.
\end{corollary}

\section{The Complexity of \DMDK{}: A Simpler Proof}
\label{sec:complexity}

In this section, we show the optimal \NExpTime{} upper bound on the complexity of \DMK{} and \DMDK{}.
As this result is not new, we begin with a brief overview of the existing literature.

\begin{itemize}[nosep]
  \item Dreben and Goldfarb~\cite{DG79} showed that the solvable Skolem class (referred to as Class~2.4 therein) has the finite model property, with a doubly exponential bound on the size of constructed models.
  Later, Goldfarb~\cite{Gol93} provided an alternative proof by extending the probabilistic method of Gurevich and Shelah~\cite{GS83} for the G\"odel class.\footnotemark{}
  However, he proved only that, as the domain size tends to infinity,
  the probability that a random structure satisfies a given sentence converges to~$1$.
  A closer inspection of the argument yields at best a doubly exponential upper bound on the required domain size.
  Moreover, none of~\cite{DG79,Gol93} addresses the complexity of satisfiability, and only a naïve \TwoNExpTime{} upper bound can be inferred.
  The optimal \NExpTime{} upper bound for the solvable Skolem class follows instead from the corresponding upper bound for \DMK{} and is stated as a corollary in~\cite{FKM24}.

  It also appears that standard proofs of the \NExpTime{} upper bound for the G\"odel class (see, e.g.,~\cite{BGG,PH23}) do not extend to the solvable Skolem class.
  In particular, these proofs exploit the fact that G\"odel-class sentences admit only two universal quantifiers.

  \item A recent work~\cite{FKM24} establishes that \DMK{} has the exponential model property and is therefore \NExpTime{}-complete.
  However, the approach is technically involved.
  In the absence of a suitable Skolemisation procedure, the authors work directly with \DMK{} sentences featuring quantifier alternation, handled via a game-theoretic approach (the so-called \emph{satisfiability game} in~\cite{FKM24}).
  Satisfiability is characterised in terms of the existence of a winning strategy in this game.
  To obtain exponential bounds, equivalence relations on game positions are introduced to yield a compact representation of strategies.
  From such representations, finite models are constructed using technically intricate arguments: construction is based on so-called paradoxical tournament graphs with vertices labelled by game positions.
\end{itemize}

\footnotetext{Goldfarb~\cite{Gol89} also developed a probabilistic proof for a different fragment, called the \emph{Maslov--Krom class}. This fragment is unrelated to the class \MK{} and should not be confused with it.}

Our proof avoids the technical complications of~\cite{FKM24} and yields the optimal \NExpTime{} upper bound via a substantially simpler argument.
Our plan is as follows.

In Subsection~\ref{subsec:prelim-reductions}, we begin with two
preliminary reductions that simplify the subsequent arguments.
In Subsection~\ref{subsec:criterion}, we provide a criterion for satisfiability of sentences in \SNF{} in terms of the existence of a combinatorial certificate, called a \emph{satisfiability witness}, and show that this condition is both necessary and sufficient.
For the sufficiency direction, we construct finite random models by refining Goldfarb's approach with a more detailed analysis of probabilities.

A purely random construction is not sufficient, as the size of satisfiability witnesses is, in general, doubly exponential.
Therefore, in Subsection~\ref{sec:exponential-certificates}, we show that every satisfiable sentence admits a witness of exponential size.
This part of the proof is inspired by~\cite{FKM24}; however, thanks to our Skolemisation procedure, ideas of~\cite{FKM24} are applied in a more direct setting.

\subsection{Preliminary reductions}
\label{subsec:prelim-reductions}

\paragraph{Standard names.}
In a structure $\fA$, distinct constants may be interpreted by the same
element, $c_1^\fA=c_2^\fA$, which is inconvenient in the subsequent
proofs.
Hence, from this point until the end of the section, we adopt the
\emph{standard name assumption}: distinct constants denote distinct
elements, and we identify each constant with its interpretation, that
is, $c^\fA=c$.
This assumption is without loss of generality for both satisfiability
and model-size bounds.
Indeed, whenever $c_1^\fA=c_2^\fA$, the constant $c_2$ can be replaced
by $c_1$ throughout $\phi$ without affecting its satisfaction or the
size of the model.
Accordingly, we partition the domain as
\(
  A=A_0\mathbin{\dot\cup}\Cons,
\)
where $A_0$ is the set of \emph{unnamed} elements and $\Cons$ is the
tuple of constants in the signature.

\paragraph{Simple assignments.}
To avoid case distinctions, decidability proofs for the G\"odel class
(see, e.g.,~\cite{GS83}) typically begin with a simple reduction that
allows one to consider only pairs of distinct elements.
Namely, an input sentence
\(
  \forall x_1,x_2\,\exists \ys\,\psi[x_1,x_2,\ys],
\)
with $\psi$ quantifier-free,
is transformed into
\tightdisplay{\[
  \forall x_1,x_2\,\exists \ys\,\exists \zs\,
  \bigl(
    x_1 \neq x_2
    \rightarrow
    \bigl(
      \psi[x_1,x_2,\ys]
      \wedge
      \psi[x_1,x_1,\zs]
    \bigr)
  \bigr),
\]}%
where $\zs$ is a tuple of fresh variables of the same length as $\ys$.
Over domains of size at least two, this transformation preserves
satisfiability.
Its advantage is that the subsequent argument need only consider the
case $x_1 \neq x_2$, since the case $x_1=x_2$ is accounted for by the second conjunct $\psi[x_1,x_1,\zs]$.
Equality is used only as part of the reduction and is not allowed
to occur in $\psi$.

We generalise this idea to $\SNF{}$ by introducing
\emph{simple} assignments, which map distinct variables to pairwise
distinct unnamed elements.
The reduction follows the same general scheme as for the G\"odel class:
it considers all possible \emph{equality patterns} among the universal
variables and constants.
Consequently, the subsequent arguments become significantly cleaner, as
it suffices to consider only simple assignments.

However, unlike in the G\"odel-class case, carrying out the reduction is
not immediate.
Indeed, with an unbounded number of universal variables, the
number of equality patterns is exponential in the sentence length.
We address this issue by carefully controlling the relevant syntactic
parameters and expressing the resulting model-size bounds in terms of
the length of the input sentence.
We formalise this convention as follows.

Given a structure $\fA$ and $m\in\N$, we write
$\langle A\rangle^m$ for the set of all $m$-tuples $\as\in A^m$
whose entries are pairwise distinct and unnamed.
The following lemma establishes a reduction that allows satisfaction of sentences in $\SNF{}$ to
be verified by considering only tuples from $\langle A\rangle^m$.

\begin{lemma}\label{lemma:reduction-to-simple-form}
  Every sentence $\phi_{\rm in}\in\SNF{}$ can be reduced to a sentence
  $\phi\in\SNF{}$ over the same signature such that,
  with the notation of Definition~\ref{def:skolem-class} applied to
  $\phi$, the following hold:
  \begin{enumerate}[nosep]
    \item\label{lemma:reduction-to-simple-form:item1}
    The sentences $\phi_{\rm in}$ and $\phi$ are logically equivalent.

    \item\label{lemma:reduction-to-simple-form:item2}
    For every structure $\fA$ with at least $m$ unnamed elements,
    \tightdisplay{\[
      \fA\models\phi
      \quad\Longleftrightarrow\quad
      \text{for every $t\in\cI$ and
      $\as\in\langle A\rangle^m$, there exists $b\in A$ such that
      $\fA\models\psi_t[\as,b]$.}
    \]}%
  \end{enumerate}
  Moreover, the reduction does not introduce new variables and the following bounds hold:
  \tightdisplay{\[
    |\cI|
      =
      2^{\,\cO(|\phi_{\rm in}|\log|\phi_{\rm in}|)},
    \qquad
    \max_{t\in\cI}|\psi_t|
      =
      \cO(|\phi_{\rm in}|),
    \qquad
    \Bigl|
      \bigcup\nolimits_{t\in\cI\colon n_t\ge2}\AAA_t
    \Bigr|
      =
      \cO(|\phi_{\rm in}|^2).
  \]}%
\end{lemma}

\begin{proof}
We fix an input sentence $\phi_{\rm in}\in\SNF{}$, and let
$\Sigma=\langle\Rels,\Cons\rangle$ be its signature.
For the duration of this proof, the notation of
Definition~\ref{def:skolem-class} refers to $\phi_{\rm in}$, rather
than to the resulting sentence $\phi$ as in the statement of the
lemma.
Hence,
\tightdisplay{\[
  \phi \;=\; \bigwedge\nolimits_{t \in \cI} \forall \xs\; \exists x_{m+1}\; \psi_t[\xs,x_{m+1}].
\]}%

Let $e$ be an equivalence relation on $\xs\cup\Cons$.
We call $e$ an \emph{equality pattern} if no two distinct constants
are $e$-equivalent, and write $\cE$ for the set of all
equality patterns.

For $e\in\cE$ and $a \in \xs \cup \Cons$, we define the \emph{$e$-representative} of $a$ as follows.
If the $e$-class $[a]_e$ contains a constant, then this constant is unique, and we take it as the $e$-representative of $a$.
Otherwise, we take the variable $x_i\in[a]_e$ with the least index~$i$.
We denote by $\rho_e$ the substitution that maps each variable in $\xs$ to its $e$-representative and fixes $x_{m+1}$.

For $t\in\cI$ and $e\in\cE$, let $\psi_{t,e}$ be obtained from
$\psi_t$ by applying $\rho_e$.
We define
\tightdisplay{\[
  \phi
  \;\coloneq\;
  \bigwedge\nolimits_{(t,e)\in\cI\times\cE}
  \forall\xs\;
  \exists x_{m+1}\;
  \psi_{t,e}[\xs,x_{m+1}].
\]}%

We claim that $\phi\in\DMDK{}$.
To see this, fix $t\in\cI$ and $e\in\cE$.
Write $\phi_{t,e} \coloneq \forall\xs\; \exists x_{m+1}\; \psi_{t,e}$.
Recall that $x_1,\dots,x_{n_t}$ are the uniform variables of
$\phi_t$.
Consider those among them that are their own $e$-representatives, that
is, the variables $x_i$ such that
$i \le n_t$
and
\(
  \rho_e(x_i)=x_i.
\)
If there are at least two such variables, select them as the uniform variables of $\phi_{t,e}$.
Otherwise, take $\phi_{t,e}$ to have no uniform variables.
Let $n_{t,e}$ denote the number of the selected uniform variables.
This selection does not necessarily satisfy the convention
imposed in Definition~\ref{def:skolem-class}, since it
needs not form an initial segment of $x_1,\dots,x_m$.
This can be remedied, when necessary, by suitably renaming the
universally quantified variables in $\phi_{t,e}$.

We now define the partition
\(
  \Atoms(\psi_{t,e})
  =
  \AAA_{t,e}\cup\BBB_{t,e}\cup\CCC_{t,e}.
\)
If $n_{t,e}=0$, we set
\tightdisplay{\[
  \AAA_{t,e}
  \coloneq
  \rho_e(\AAA_t)\cup\rho_e(\BBB_t),
  \qquad
  \BBB_{t,e}\coloneq\emptyset,
  \qquad
  \CCC_{t,e}\coloneq\rho_e(\CCC_t).
\]}%
If $n_{t,e}\ge2$, we instead set
\tightdisplay{\[
  \AAA_{t,e}\coloneq\rho_e(\AAA_t),
  \qquad
  \BBB_{t,e}\coloneq\rho_e(\BBB_t),
  \qquad
  \CCC_{t,e}\coloneq\rho_e(\CCC_t).
\]}%

In both cases, every atom in $\AAA_{t,e}$ contains at most one
variable, every atom in $\BBB_{t,e}$ contains precisely the declared
uniform variables, and every atom in $\CCC_{t,e}$ contains the
existential variable $x_{m+1}$.
Thus, every conjunct of $\phi$ belongs to $\DMK{}$, and hence
$\phi\in\DMDK{}$.

It is straightforward to verify that $\phi_{\rm in}$ and $\phi$ are logically equivalent.
Indeed, since the variables in $\xs$ are universally quantified,
$\phi_{\rm in}$ already accounts for all possible identifications among variables and substitutions by constants.
The constructed sentence $\phi$ merely makes this fact explicit through considering all equality patterns.

This explicit treatment of equality patterns allows satisfaction of
$\phi$ over any structure $\fA$ with at least $m$ unnamed elements to
be verified by considering only tuples from $\langle A\rangle^m$.
To see this, take a tuple $\as \in A^m$ and write $\as = \langle a_1,\dots,a_m \rangle$.
Let $e$ be the equivalence relation generated by
\tightdisplay{\[
\bigl\{ (x_i,x_j) \mid i,j \in [m],\, a_i=a_j \bigr\}
\;\cup\;
\bigl\{ (x_i,c) \mid i \in [m],\, c \in \Cons,\, a_i=c \bigr\}.
\]}%
Distinct constants cannot be identified because we work under the
standard name assumption.
Hence, $e$ is an equality pattern.
We construct a tuple $\as'\in\langle A\rangle^m$ in two steps.
First, for every $e$-representative variable $x_i$, its $e$-class
contains no constant, and we set the $i$th entry of $\as'$ as in $\as$.
These entries are pairwise distinct and unnamed.
Second, using the assumption that $\fA$ has at least $m$ unnamed
elements, we fill the remaining positions with fresh unnamed elements
so that all entries of $\as'$ are pairwise distinct.
The definition of $\rho_e$ then ensures that
\tightdisplay{\[
  \fA \models \psi_{t,e}[\as', b]
  \quad\Longleftrightarrow\quad
  \fA \models \psi_t[\as, b]
  \qquad
  \text{for every $b \in A$.}
\]}%
Thus, witnesses for all tuples in $\langle A\rangle^m$ suffice to 
provide witnesses for arbitrary tuples.

It remains to bound the increase in the relevant syntactic parameters.
The first two bounds of the lemma are immediate.
The number of conjuncts of $\phi$ is
\(
  |\cI\times\cE|
  =
  2^{\cO(|\phi_{\rm in}|\log|\phi_{\rm in}|)},
\)
and every introduced conjunct satisfies
\(
  |\psi_{t,e}|=\cO(|\phi_{\rm in}|).
\)
We now establish the third bound, namely, that the total number of
distinct $\AAA$-atoms occurring in conjuncts of uniformity at least~$2$
remains quadratic in the size of the input sentence.

Fix $t\in\cI$.
For every $e\in\cE$ with $n_{t,e}\ge2$, we have
\(
  \AAA_{t,e}=\rho_e(\AAA_t).
\)
It therefore suffices to bound the number of distinct images of each
atom $\alpha\in\AAA_t$ under the substitutions $\rho_e$.

Suppose that $\Vars(\alpha)\subseteq\{x_i\}$ for some $i\in[m]$.
For every $e \in \cE$, the substitution $\rho_e$ maps $x_i$
either to some variable $x_j$, with $j\in[m]$, or to a constant in
$\Cons$.
Hence, as $e$ ranges over $\cE$, the atom $\alpha$ has at most
$m+|\Cons|$ distinct images.

Consequently,
\tightdisplay{\[
  \Bigg|
    \bigcup_{\substack{(t,e)\in\cI\times\cE\\ n_{t,e}\ge2}}
    \AAA_{t,e}
  \Bigg|
  \le
  \sum_{t\in\cI}
    |\AAA_t|\cdot\bigl(m+|\Cons|\bigr).
\]}%
Since
\(
  \sum_{t\in\cI}|\AAA_t|
  =
  \cO(|\phi_{\rm in}|)
\)
and
\(
  m+|\Cons|
  =
  \cO(|\phi_{\rm in}|),
\)
the desired quadratic bound follows.
\end{proof}

\subsection{A Satisfiability Criterion}
\label{subsec:criterion}

For the rest of this section, we fix an \emph{input sentence}
$\phi_{\rm in}\in\SNF{}$ and let $\phi$ be the sentence obtained from
$\phi_{\rm in}$ by Lemma~\ref{lemma:reduction-to-simple-form}.
Let $\Sigma=\langle\Rels,\Cons\rangle$ denote their common signature.
Unless stated otherwise, we consider only $\Sigma$-structures with at
least $m$ unnamed elements.
From this point on, the notation of Definition~\ref{def:skolem-class}
always refers to the components of $\phi$, while
$\phi_{\rm in}$ is retained only to record the size of the input.
This distinction is important because the reduction incurs an
exponential blow-up, whereas the model-size bounds established below
must be stated in terms of $|\phi_{\rm in}|$ rather than $|\phi|$.
For this purpose, Lemma~\ref{lemma:reduction-to-simple-form} also
records precise bounds on the relevant syntactic parameters of $\phi$.
The significance of these bounds will become clear in
Subsection~\ref{sec:exponential-certificates}.

\paragraph{Satisfiability witnesses.}
Following~\cite{Gol93}, we characterise satisfiability of $\phi$ in terms of the existence of a collection of outer-types satisfying suitable realisability and extension properties.

We first characterise,
in combinatorial terms, when a collection of outer-types corresponds to the set of all $k$-outer-types for $k \le m+1$ realised in a structure in which every realised $1$-type has at least $m+1$ distinct realisations.
This is captured by the notion of $\phi$-realisability.

\begin{definition}\label{definition:phi-realisability}
  A set of outer-types $\cS \subseteq \TTTo(\Sigma)$ is \emph{$\phi$-realisable} if the following conditions hold:
  \begin{itemize}[nosep]
    \item $\cS$ contains a unique $0$-type, and every other outer-type $\tau_\ro \in \cS$ is consistent with it.
    
    \item $\cS$ contains at least one $1$-type.

    \item $\cS$ is closed under automorphisms of outer-types induced by permutations of variables.

    \item For every $k \in [m+1]$ and every $k$-outer-type $\tau_\ro \in \cS$,
    $\tau_\ro \langle x_i \rangle \in \cS$ for every $i \in [k]$.

    \item For every $k \in [m+1]$ and all $1$-types
    $\pi_1,\dots,\pi_k \in \cS$, there exists a $k$-outer-type
    $\tau_\ro \in \cS$ such that $\tau_\ro \langle x_i \rangle = \pi_i$ for every $i \in [k]$.
  \end{itemize}
\end{definition}

Suppose for a moment that we want to construct a model of $\phi$ using
only outer-types from a fixed $\phi$-realisable set $\cS$.
A natural approach is to proceed inductively as follows.

Suppose that a structure $\fA$ has already been constructed.
We select a tuple $\as \in \langle A \rangle^m$,
representing an assignment to the universal variables, and an index
$t \in \cI$ of a conjunct that is unsatisfied for $\as$, that is,
\(
  \fA\not\models
  \exists x_{m+1}\; \psi_t[\as,x_{m+1}].
\)
We then extend $\fA$ in three steps.
First, we introduce a fresh element $b$ to serve as a witness.
Next, we define the $(m+1)$-type of the extended tuple $\as b$ so that
\(
  \fA\models\psi_t[\as,b].
\)
Finally, using $\phi$-realisability of $\cS$,
we complete the structure by choosing
outer-types from $\cS$ for the remaining tuples involving $b$.

For such a construction to succeed, two properties must be satisfied.
First, the chosen $(m+1)$-type must be compatible both with the
$m$-type already realised by $\as$ and with the collection $\cS$ of
allowed outer-types.
Second, it must be possible to assign outer-types from $\cS$ to the remaining
tuples involving $b$ consistently with the facts already defined.

We defer the second issue and focus for now on capturing the first at
the level of types.

An \emph{$\cS$-requirement} is any $m$-type $\tau$ that induces only outer-types from $\cS$.
An \emph{$(\cS,t)$-extension} of $\tau$ is any $(m+1)$-type $\sigma$ that extends $\tau$, satisfies $\psi_t$, and likewise induces only outer-types from $\cS$.
The \emph{$\phi$-extension property} requires that every $\cS$-requirement has an $(\cS,t)$-extension for every index $t \in \cI$.
Formally:

\begin{definition}\label{definition:phi-extension}
  Let $\cS \subseteq \TTTo(\Sigma)$.
  An $m$-type $\tau$ is an \emph{$\cS$-requirement} if
  $\TTTo(\tau) \subseteq \cS$.
  Given $t \in \cI$, an \emph{$(\cS,t)$-extension} of $\tau$
  is an $(m{+}1)$-type $\sigma$ such that
  $\TTTo(\sigma) \subseteq \cS$ and
  $\sigma \models \tau \wedge \psi_t$.
  
  The set $\cS$ has the \emph{$\phi$-extension property}
  if, for every $\cS$-requirement $\tau$ and every $t \in \cI$,
  there exists an $(\cS,t)$-extension of $\tau$.
\end{definition}

We combine the notions of $\phi$-realisability and the $\phi$-extension
property in the following definition of a satisfiability witness.
As we show subsequently, $\phi$ is satisfiable if and only if there
exists a collection of outer-types with these two properties.

\begin{definition}\label{def:satisfiability-witness}
  A set of outer-types $\cW \subseteq \TTTo(\Sigma)$ is called a
  \emph{satisfiability witness} for $\phi$ if it is $\phi$-realisable
  and has the $\phi$-extension property.
\end{definition}

\paragraph{Completeness.}
The notion of a satisfiability witness is \emph{complete}: whenever
$\phi$ is satisfiable, it admits a satisfiability witness.
This is established by the following argument.

We begin with an auxiliary claim describing a type-surgery operation
that allows us to transfer witnessing extensions between $m$-types
while controlling the induced outer-types.

\begin{claim}\label{claim:type-surgery}
  Fix $t\in\cI$, and let $\tau_1,\tau_2$ be $m$-types such that
  \begin{enumerate}[nosep]
    \item
    $\tau_1\models\alpha\iff\tau_2\models\alpha$ for every
    $\alpha\in\AAA_t\cup\BBB_t$;
    \item
    $\tau_1 \langle x_i \rangle= \tau_2 \langle x_i \rangle$ for every $i\in[m]$.
  \end{enumerate}
  Let $\sigma_1$ be an $(m+1)$-type extending $\tau_1$ and satisfying
  $\psi_t$, and define
  \(
    \sigma_2\coloneq(\sigma_1\setminus\tau_1)\cup\tau_2.
  \)
  Then $\sigma_2$ is an $(m+1)$-type extending $\tau_2$ and satisfying
  $\psi_t$.
  Moreover,
  \(
    \TTTo(\sigma_2)
    \subseteq
    \TTTo(\tau_2)\cup\TTTo(\sigma_1).
  \)
\end{claim}

\begin{proof}
  Clearly, $\sigma_2$ extends $\tau_2$.
  We now verify that $\sigma_2\models\psi_t$.
  Replacing $\tau_1$ by $\tau_2$ preserves the truth
  values of all atoms in $\AAA_t\cup\BBB_t$.
  Moreover, since every atom in $\CCC_t$ contains $x_{m+1}$, its truth
  value is determined by the literals in $\sigma_1\setminus\tau_1$ and
  is therefore the same in $\sigma_2$ as in $\sigma_1$.
  Thus, every atom of $\psi_t$ has the same truth value in $\sigma_1$
  and $\sigma_2$.
  Since $\sigma_1\models\psi_t$, it follows that
  $\sigma_2\models\psi_t$.

  It remains to verify the inclusion of induced outer-types.
  Since $\tau_1$ and $\tau_2$ agree on the induced $1$-types, and the
  literals in $\sigma_1\setminus\tau_1$ are retained in $\sigma_2$,
  every outer-type induced by $\sigma_2$ that involves $x_{m+1}$ is
  also induced by $\sigma_1$.
  Every remaining outer-type induced by $\sigma_2$ is induced on some of
  the first $m$ variables, and hence is induced by $\tau_2$.
  Therefore,
  \(
    \TTTo(\sigma_2)
    \subseteq
    \TTTo(\tau_2)\cup\TTTo(\sigma_1).
  \)
\end{proof}

We now return to the completeness argument.
Let $\fA\models\phi$.
We first transform $\fA$ into a new structure $\fB$ by taking its
\emph{infinite unfolding}, constructed as follows.
The domain $B$ consists of countably many copies of each element of
$A$ and is identified with $A\times\N$, where $(c,0)$ is identified
with the constant $c\in\Cons$.
For every relation symbol $R\in\Rels$ of arity $k$, we define
\tightdisplay{\[
  R^\fB
  \;\coloneq\;
  \bigl\{\bigl\langle(a_1,i_1),\dots,(a_k,i_k)\bigr\rangle \mid a_1,\dots,a_k \in A,\, i_1,\dots,i_k \in \N,\, R^\fA(a_1,\dots,a_k)\bigr\}.
\]}%
By a standard argument, $\fB\models\phi$.
Moreover, for every $t\in\cI$ and every
$\as\in\langle B\rangle^m$, there exists an unnamed element
$b_{\as,t}\in B$ not occurring in $\as$ such that
\(
  \fB\models\psi_t[\as,b_{\as,t}].
\)
Indeed, every element $(a,i)\in B$, including one identified with a
constant, is indistinguishable from its copy $(a,j) \in B$.
Hence, any witness for $\psi_t$ can be replaced by an unnamed copy not
occurring in $\as$.

Let $\cW$ consist of all $k$-outer-types, for $k\le m+1$, realised
in $\fB$.
We claim that $\cW$ is a satisfiability witness for $\phi$.

By construction of $\fB$, every realised $1$-type has
infinitely many distinct realisations.
It directly follows that $\cW$ is $\phi$-realisable.
It remains to verify the $\phi$-extension property.

Fix $t\in\cI$ and let $\tau$ be a $\cW$-requirement.
Consider the atoms in $\AAA_t \cup \BBB_t$, that is, those that contain only universal variables.
The atoms in $\AAA_t$ depend only on the $1$-types induced by $\tau$,
while those in $\BBB_t$ depend only on the outer-type induced on the uniform
variables $x_1,\dots,x_{n_t}$.
Hence, although $\tau$ itself need not be realised in $\fB$, we can
choose a tuple $\as\in\langle B\rangle^m$ that agrees with $\tau$ on
$\AAA_t\cup\BBB_t$.
Indeed, we first choose distinct unnamed elements
$a_1,\dots,a_{n_t}$ realising
$\tau\langle x_1,\dots,x_{n_t} \rangle$.
We then choose distinct unnamed elements
$a_{n_t+1},\dots,a_m$, also distinct from the preceding ones, such that
\(
  \tp{\fB}[a_i]=\tau\langle x_i \rangle
\)
for every $i\in[n_t+1,m]$.

Let $\tau'\coloneq\tp{\fB}[\as]$ and $\sigma'\coloneq\tp{\fB}[\as,\bs_{\as,t}]$.
Then $\sigma'$ extends $\tau'$ and satisfies $\psi_t$.
By the choice of $\as$, the types $\tau$ and $\tau'$
agree on all atoms in $\AAA_t\cup\BBB_t$ and induce the same
sequences of $1$-types.
Hence, Claim~\ref{claim:type-surgery}, applied to $\tau$ and $\tau'$,
shows that
\(
  \sigma\coloneq(\sigma'\setminus\tau')\cup\tau
\)
extends $\tau$, satisfies $\psi_t$, and
\(
  \TTTo(\sigma)
  \subseteq
  \TTTo(\tau)\cup\TTTo(\sigma').
\)
Since $\tau$ is a $\cW$-requirement, we have
$\TTTo(\tau)\subseteq\cW$.
Moreover, $\sigma'$ is realised in $\fB$, and hence
$\TTTo(\sigma')\subseteq\cW$.
Therefore,
\(
  \TTTo(\sigma)\subseteq\cW,
\)
so $\sigma$ is a $(\cW,t)$-extension of $\tau$.

Consequently, $\cW$ has the $\phi$-extension property and is therefore
a satisfiability witness for~$\phi$.

\paragraph{Random models.}
We now establish the more interesting direction, namely that the notion
of a satisfiability witness is \emph{sound} for finite satisfiability:
whenever $\phi$ admits a satisfiability witness, it has a finite model.
The proof follows Goldfarb's probabilistic approach~\cite{Gol93}, but refines the
probability analysis to obtain an explicit bound on the size of the constructed models.

The randomised construction relies on a simple observation, which we
call the \emph{independence property for outer-types}.
Once the $1$-types of the elements have been fixed, outer-types can be
assigned independently to distinct sets of elements, provided that they
are consistent with these $1$-types.
For each such set, it suffices to choose the outer-type of a single
ordering; those of all other orderings are then determined by the
corresponding permutations of variables.

This independence property fails for ordinary types.
To see this, consider
$\as=\langle a_1,a_2,a_3\rangle$ and
$\as'=\langle a_2,a_3,a_4\rangle$.
Even after the $1$-types of the elements have been fixed, the
$3$-types of $\as$ and $\as'$ cannot be chosen independently:
the literals $\pm R(x_2,x_3)$ in the type of $\as$ and
$\pm R(x_1,x_2)$ in the type of $\as'$ both determine the truth value
of $R(a_2,a_3)$ and must therefore agree.
By contrast, $3$-outer-types avoid this dependence
by omitting such literals with two distinct variables.

We now proceed to the construction.
Fix a satisfiability witness $\cW$ for $\phi$.
Let $p \in \N$ be a parameter, to be determined later, specifying the number of unnamed elements in the domain.
We construct a finite $\Sigma$-structure $\fA$ with domain \( A \coloneq [p] \cup \Cons \) in three steps: 
\begin{enumerate}[nosep]
  \item\label{item:random-model-stage1}
  For every $a \in [p]$, choose uniformly at random a
  $1$-type $\pi \in \cW \cap \TTT^1(\Sigma)$, and set
  \( \tp{\fA}[a] \coloneq \pi. \)

  \item\label{item:random-model-stage2}
  For every $k \in [2,m+1]$ and every
  $\{a_1 < \dots < a_k\} \subseteq [p]$, choose uniformly at random a
  $k$-outer-type
  $\tau_\ro \in \cW \cap \TTT^k_\ro(\Sigma)$ such that
  \( \tau_\ro\langle x_i \rangle = \tp{\fA}[a_i] \) for every $i \in [k]$.
  Since $\cW$ is $\phi$-realisable, at least one such $k$-outer-type exists.
  By the independence property for outer-types,
  $\tau_\ro$ can
  be assigned to $\langle a_1,\dots,a_k\rangle$ without affecting the
  outer-type assigned to any other tuple with a different set of underlying elements.
  Thus, set \( \otp{\fA}[a_1,\dots,a_k] \coloneq \tau_\ro. \)

  \item Set to false the truth values of all facts with more than $m+1$ distinct unnamed elements.
\end{enumerate}

The construction of $\fA$ is now complete.
In particular, recall that types and outer-types also include literals involving constants.
Thus, all facts involving them are implicitly defined.

We now show that $\bbP[\fA \models \phi] > 0$ once the parameter $p$
is sufficiently large.
For the proof, let
$s_\gamma \coloneq \max_{t \in \cI} |\CCC_t|$
and
$\delta \coloneq |\cW|^{-(s_\gamma+1)}$.

\begin{claim}\label{claim:extension-probability}
  Let $t \in \cI$ and let $\as \in \langle A \rangle^m$.
  Then
  \(
    \bbP\big[
      \fA \not\models
      \exists x_{m+1}\,\psi_t[\as,x_{m+1}]
    \big]
    \le
    (1-\delta)^{p-m}.
  \)
\end{claim}

\begin{proof}
  Write $\as=\langle a_1,\dots,a_m\rangle$ and condition on the type
  \(
    \tau\coloneq\tp{\fA}[\as].
  \)
  By construction of $\fA$, every outer-type induced by $\tau$ belongs
  to $\cW$, and hence $\tau$ is a $\cW$-requirement.
  Since $\cW$ is a satisfiability witness, $\tau$ admits a
  $(\cW,t)$-extension $\sigma$.

  Fix an element $a_{m+1}\in[p]\setminus\as$.
  We lower-bound the probability that $a_{m+1}$ is a witness for $\as$,
  that is, that
  \(
    \fA\models\psi_t[\as,a_{m+1}].
  \)
  Since $\sigma\models\psi_t$, it suffices that every atom of $\psi_t$
  has the same truth value on $\langle a_1,\dots,a_{m+1}\rangle$ as in
  $\sigma$, namely,
  \tightdisplay{\begin{equation}\label{eqn:claim:extension-probability-1}
    \tp{\fA}[a_1,\dots,a_{m+1}]\models\alpha
    \quad\iff\quad
    \sigma\models\alpha
    \qquad
    \text{for every $\alpha\in\Atoms(\psi_t)$.}
  \end{equation}}%

  Since $\sigma$ extends $\tau$, the atoms in
  $\AAA_t\cup\BBB_t$ already have the required truth values.
  It therefore remains to consider the atoms in $\CCC_t$.
  For an atom $\alpha\in\CCC_t$ with
  \(
    \Vars(\alpha)
    =
    \{x_{i_1},\dots,x_{i_k},x_{m+1}\},
  \)
  the required truth value is ensured whenever
  \tightdisplay{\begin{equation}\label{eqn:claim:extension-probability-2}
    \otp{\fA}[a_{i_1},\dots,a_{i_k},a_{m+1}]
    =
    \sigma\langle x_{i_1},\dots,x_{i_k},x_{m+1} \rangle.
  \end{equation}}%

  Let $\cE(a_{m+1})$ be the event that
  \eqref{eqn:claim:extension-probability-2} holds for every
  $\alpha \in \CCC_t$.
  For this event to occur, the $1$-type of $a_{m+1}$ must agree with the
  $1$-type induced by $\sigma$ on $x_{m+1}$, and the outer-type of every
  relevant tuple
  \(
    \langle a_{i_1},\dots,a_{i_k},a_{m+1}\rangle
  \)
  must agree with the corresponding outer-type induced by $\sigma$.

  Each of the above requirements is an admissible choice in the random
  construction, and each of them is made uniformly from a set of at most
  $|\cW|$ possibilities.
  As there are at most $|\CCC_t|+1 \le s_\gamma+1$ such requirements,
  the product rule for probabilities gives
  \tightdisplay{\begin{equation}\label{eqn:claim:extension-probability-3}
    \bbP\big[
      \cE(a_{m+1})
      \mid
      \tp{\fA}[\as]=\tau
    \big]
    \ge
    |\cW|^{-(s_\gamma+1)}
    =
    \delta.
  \end{equation}}%

  Now let
  \(
    B \coloneq [p] \setminus \as.
  \)
  The events $\{\cE(b) \mid b \in B\}$ are mutually independent,
  conditional on $\tp{\fA}[\as]=\tau$.
  Indeed, for distinct $b,b' \in B$, the events $\cE(b)$ and $\cE(b')$
  depend on disjoint random choices, except for the $1$-types of the
  elements of $\as$, which are fixed by the conditioning.
  Consequently, the probability that no element $b \in B$ provides a witness for $\as$ is at most
  \tightdisplay{\begin{equation}\label{eqn:claim:extension-probability-4}
      \bbP\big[
        \bigcap\nolimits_{b\in B} \neg\cE(b)
        \mid
        \tp{\fA}[\as]=\tau
      \big]
      =
      \prod\nolimits_{b\in B}
      \bbP\big[
        \neg\cE(b)
        \mid
        \tp{\fA}[\as]=\tau
      \big]
      \le
      (1-\delta)^{|B|}.
  \end{equation}}%

  Since $|B|=p-m$, we obtain
  \(
    \bbP\big[
      \fA \not\models
      \exists x_{m+1}\,\psi_t[\as,x_{m+1}]
      \mid
      \tp{\fA}[\as]=\tau
    \big]
    \le
    (1-\delta)^{p-m}.
  \)
  As the final bound is independent of $\tau$, the condition $\tp{\fA}[\as]=\tau$ may be dropped.
\end{proof}

We now bound the probability that the sampled structure fails to satisfy $\phi$.
By the form of $\phi$, we have $\fA \models \phi$ precisely when
\(
  \fA \models
  \exists x_{m+1}\,\psi_t[\as,x_{m+1}]
\)
for every $t \in \cI$ and every tuple $\as \in \langle A \rangle^m$.
For every such pair $t$ and $\as$,
Claim~\ref{claim:extension-probability} gives
\(
  \bbP\big[
    \fA \not\models
    \exists x_{m+1}\,\psi_t[\as,x_{m+1}]
  \big]
  \le
  (1-\delta)^{p-m}.
\)
Taking the union bound over all $t \in \cI$ and
$\as \in \langle A \rangle^m$, and using $|\langle A \rangle^m| \le p^m$, we obtain
\tightdisplay{\begin{equation}\label{eqn:union-bound:1}
  \bbP[\fA \not\models \phi]
  \le
  \sum\nolimits_{t \in \cI} \sum\nolimits_{\as \in \langle A \rangle^m}
  \bbP\big[
    \fA \not\models
    \exists x_{m+1}\,\psi_t[\as,x_{m+1}]
  \big]
  \le
  |\cI| \cdot p^m \cdot (1-\delta)^{p-m}.
\end{equation}}%

For fixed $\phi$ and $\cW$, the right-hand side of
\eqref{eqn:union-bound:1} tends to $0$ as $p \to \infty$, since its
polynomial factor $p^m$ is dominated by the exponentially decreasing
factor $(1-\delta)^{p-m}$.
Thus, for sufficiently large $p$,
\(
  \bbP[\fA \models \phi] > 0.
\)
Consequently, there exists a finite model of $\phi$.

A more careful analysis of~\eqref{eqn:union-bound:1} gives the following
quantitative version of the result, which will be needed for our
complexity bounds.

\begin{lemma}\label{lemma:random-models}
  Let $\phi_{\rm in}\in\SNF{}$, and let $\phi$ be obtained from
  $\phi_{\rm in}$ by the reduction of
  Lemma~\ref{lemma:reduction-to-simple-form}.
  Given a satisfiability witness $\cW$ for $\phi$, the sentence $\phi$
  admits a finite model $\fA$ with domain size
  \tightdisplay{\[
    |A|
    =
    2^{\cO(|\phi_{\rm in}|\log|\phi_{\rm in}|)}
    \cdot
    |\cW|^{\cO(|\phi_{\rm in}|)}.
  \]}%
\end{lemma}

\begin{proof}
  It suffices to choose $p$ so that the right-hand side
  of~\eqref{eqn:union-bound:1} is strictly smaller than $1$.
  For convenience, we impose the stronger requirement that it be at most
  $1/e$, that is, we require
  \tightdisplay{\begin{equation}\label{eqn:union-bound:2}
    |\cI|\cdot p^m\cdot(1-\delta)^{p-m}
    \le \frac{1}{e}.
  \end{equation}}%
  
  Taking natural logarithms, inequality~\eqref{eqn:union-bound:2} is equivalent to
  \tightdisplay{\begin{equation}\label{eqn:union-bound:3a}
    \log|\cI|
    +m\cdot\log p
    +(p-m)\cdot\log(1-\delta)
    \le -1.
  \end{equation}}%
  Using
  \(
    \log(1-x)\le -x
  \)
  for $x\in(0,1)$, it suffices that
  \tightdisplay{\begin{equation}\label{eqn:union-bound:3}
    \log|\cI|
    +m\cdot\log p
    -(p-m)\cdot\delta
    \le -1.
  \end{equation}}%
  Rearranging, this is equivalent to
  \tightdisplay{\begin{equation}\label{eqn:union-bound:4}
    \bigl(\log|\cI|+m\cdot\log p+1\bigr)\cdot\delta^{-1}
    +m
    \le p.
  \end{equation}}%
  Bounding each of $\log|\cI|$, $m$, and $1$ by $|\phi|$, it suffices
  that
  \tightdisplay{\begin{equation}\label{eqn:union-bound:5}
    \bigl(|\phi|+|\phi|\cdot\log p+|\phi|\bigr)\cdot\delta^{-1}
    +|\phi|
    \le p.
  \end{equation}}%
  Absorbing the lower-order terms, it is enough to require
  \tightdisplay{\begin{equation}\label{eqn:union-bound:6}
    4 \cdot |\phi| \cdot \delta^{-1} \cdot \log p
    \le p.
  \end{equation}}%
  Since
  \(
    \log p\le \sqrt{p}
  \)
  for all $p \ge 1$, inequality~\eqref{eqn:union-bound:6} is satisfied
  whenever
  \tightdisplay{\begin{equation}\label{eqn:union-bound:7}
    4 \cdot |\phi| \cdot \delta^{-1}
    \le \sqrt{p}.
  \end{equation}}%
  Thus, it suffices to take
  \(
    p
    \ge
    16 \cdot |\phi|^2 \cdot \delta^{-2}.
  \)

  By Lemma~\ref{lemma:reduction-to-simple-form},
  \(
    |\phi|
    =
    2^{\cO(|\phi_{\rm in}|\log|\phi_{\rm in}|)}
  \)
  and
  \(
    \delta^{-1}
    =
    |\cW|^{\cO(|\phi_{\rm in}|)}.
  \)
  Consequently, for
  \tightdisplay{\begin{equation}
    p
    =
    2^{\cO(|\phi_{\rm in}|\log|\phi_{\rm in}|)}
    \cdot
    |\cW|^{\cO(|\phi_{\rm in}|)},
  \end{equation}}%
  the random model construction yields a model $\fA\models\phi$
  of the claimed size.
\end{proof}

\subsection{Exponential-Size Certificates}
\label{sec:exponential-certificates}

We retain the notation from Subsection~\ref{subsec:criterion}: the
input sentence $\phi_{\rm in}\in\SNF$, its conversion $\phi$ given by
Lemma~\ref{lemma:reduction-to-simple-form}, and the signature
$\Sigma=\langle\bar{R},\bar{c}\rangle$.
In this subsection, using techniques from~\cite{FKM24}, we show that
whenever $\phi$ is satisfiable, it admits a satisfiability witness of
size exponential in $|\phi_{\rm in}|$.
Stating the bound in terms of $\phi_{\rm in}$ is important,
because $\phi$ is exponentially larger than~$\phi_{\rm in}$.

Obtaining the desired bound is not immediate.
A straightforward estimate on the number of possible $k$-outer-types
gives a doubly exponential bound in the size of the signature $\Sigma$.
In the presence of constants and relation symbols of unbounded arity,
this bound is already attained at $k = 1$.
Thus, simply counting all outer-types that may occur in a satisfiability
witness yields only a doubly exponential bound.
To obtain a singly exponential bound, we act as follows.

We start with an arbitrary satisfiability witness $\cW$ for $\phi$ and reduce its size in two stages.
First, we reduce the number of $1$-types occurring in $\cW$, obtaining a new witness $\cW'$.
This step may, in fact, increase the number of $k$-outer-types for $k > 1$.
Second, we take the witness $\cW'$ and reduce the number of $k$-outer-types for $k > 1$ while preserving the number of $1$-types. This way we obtain the final witness $\cW''$ which is exponentially bounded in $|\phi_{\rm in}|$.

\paragraph{Reduction of $1$-types.}
Recall that the notation of Definition~\ref{def:skolem-class} refers to
$\phi$, and, in particular, that
\( \Atoms(\psi_t)=\AAA_t\cup\BBB_t\cup\CCC_t. \)
For the construction, we define the auxiliary sets
\tightdisplay{\begin{equation}\label{eqn:AAA-star}
  \AAA^*
  \;\coloneq\;
  \bigcup\nolimits_{t\in\cI\colon n_t\ge2}\AAA_t,
  \qquad
  \CCC'_t
  \;\coloneq\;
  \bigl\{
    \alpha\in\CCC_t
    \mid
    \Vars(\alpha)=\{x_{m+1}\}
  \bigr\}.
\end{equation}}%

For brevity, we introduce the following notation for the relevant set sizes:
\tightdisplay{\[
s^*_\alpha = |\AAA^*|,\qquad
s_\alpha = \max\nolimits_{t \in \cI} |\AAA_t|,\qquad
s_\beta = \max\nolimits_{t \in \cI} |\BBB_t|,\qquad
s'_\gamma = \max\nolimits_{t \in \cI} |\CCC'_t|.
\]}%

We reduce the number of $1$-types occurring in $\cW$ by introducing a family of equivalence relations.
For $t \in \cI$ and $1$-types $\pi_1,\pi_2 \in \cW \cap \TTT^1(\Sigma)$, define $\pi_1 \sim_t \pi_2$ if
\tightdisplay{\begin{equation*}
  \pi_1 \models \alpha[x_1/\,\cdot\,] \;\iff\; \pi_2 \models \alpha[x_1/\,\cdot\,] \quad\text{for every $\alpha \in \AAA^* \cup \CCC'_t $}.
\end{equation*}}%

Above, the notation $[x_1/\,\cdot\,]$ stands for the substitution replacing every variable with $x_1$.

Let $N_{\sim_t}$ denote the number of equivalence classes of $\sim_t$. By definition,
\tightdisplay{\begin{equation}\label{eqn:size-estimate-1}
  \sum\nolimits_{t \in \cI} N_{\sim_t} \;\le\; |\cI| \cdot 2^{s^*_\alpha + s'_\gamma}.
\end{equation}}%

For every $t\in\cI$, we fix a choice function
\(
  \chc_t
  \colon
  (\cW\cap\TTT^1(\Sigma))/_{\sim_t}
  \longrightarrow
  \cW\cap\TTT^1(\Sigma)
\)
that selects a representative from each equivalence class of $\sim_t$.
For convenience, for a $1$-type $\pi\in\cW$, we write
\(
  \chc(\pi;t)
  \coloneq
  \chc_t([\pi]_{\sim_t})
\)
for the chosen $\sim_t$-representative of $\pi$.

More generally, let $\tau_\ro\in\cW$ be a $k$-outer-type and let
$t_1,\dots,t_k\in\cI$.
For every $i\in[k]$, let
$\pi_i\coloneq \tau_\ro\langle x_i \rangle$
and
$\pi'_i\coloneq\chc(\pi_i;t_i)$.
We write
\(
  \chc(\tau_\ro;t_1,\dots,t_k)
\)
for the $k$-outer-type obtained from $\tau_\ro$ by replacing each induced
$1$-type $\pi_i$ by $\pi'_i$,
and leaving all other literals unchanged;
formally: 
\tightdisplay{\[
  \chc(\tau_\ro;t_1,\dots,t_k) \;\coloneq\;
  \tau_\ro \setminus
  \bigl(\pi_1[x_1] \cup \cdots \cup \pi_k[x_k]\bigr)
  \cup \pi'_1[x_1] \cup \cdots \cup \pi'_k[x_k],
\]}
where $\pi[x_i]$ denotes
the result of substituting the variable $x_1$ by $x_i$ in
the $1$-type $\pi$.

Observe that replacing the induced $1$-types leaves unchanged every
literal involving at least two distinct variables.
Moreover, since all $1$-types under consideration belong to $\cW$,
they induce the same $0$-type, and hence the variable-free literals remain unchanged as well.
Therefore, the outer-types $\tau_\ro$ and $\chc(\tau_\ro;t_1,\dots,t_k)$ may differ only in literals involving precisely one variable.

We now define
\tightdisplay{\[
  \cW'
  \;\coloneq\;
  \bigl\{
    \chc(\tau_\ro;t_1,\dots,t_k)
    \mid
    k\in[0,m+1],\
    \tau_\ro\in\cW\cap\TTT^k_\ro(\Sigma),\
    t_1,\dots,t_k\in\cI
  \bigr\}.
\]}%

We claim that the set $\cW'$ defined above is again a satisfiability
witness for $\phi$.
Before formalising and proving this claim, we introduce one auxiliary
notion.

Given a consistent set of literals $\Gamma$ over variables
$x_1,\dots,x_k$, a \emph{$\cW$-completion} of $\Gamma$ is a $k$-type
$\tau$ such that
\(
  \Gamma\subseteq\tau
\)
and
\(
  \TTTo(\tau)\subseteq\cW
\).
Whenever $\Gamma$ is obtained by consistently fixing the $1$-types
and some outer-types from $\cW$, $\phi$-realisability of $\cW$ ensures
that a $\cW$-completion exists.
Indeed, for each remaining set of variables, $\phi$-realisability
provides an outer-type from $\cW$ consistent with the prescribed
$1$-types, while the independence property for outer-types allows these
choices to be made independently of the outer-types already fixed.

\begin{claim}\label{claim:witness-small-1-types}
  The set $\cW'$ is a satisfiability witness for $\phi$.
\end{claim}

\begin{proof}
  A routine verification of the definition shows that $\cW'$ is
  $\phi$-realisable.

  To establish the $\phi$-extension property for $\cW'$, fix a
  $\cW'$-requirement $\tau'$ and $t\in\cI$.
  Write
  \tightdisplay{\[
    \text{$\pi'_i\coloneq \tau' \langle x_i \rangle$ for $i\in[m]$,}
    \qquad
    \tau'_\ro\coloneq \tau' \langle x_1,\dots,x_{n_t} \rangle.
  \]}%

  By construction of $\cW'$, we can choose an $n_t$-outer-type
  $\tau_\ro\in\cW$ and indices $t_1,\dots,t_{n_t} \in \cI$ such that
  \(
    \tau'_\ro=\chc(\tau_\ro; t_1,\dots,t_{n_t}).
  \)
  Moreover, for every $i\in[n_t+1,m]$, we can choose $t_i\in\cI$ such that
  $\pi'_i = \chc(\pi'_i; t_i)$.
  We now construct a $\cW$-requirement $\tau$ in three steps:
  \begin{enumerate}[nosep]
    \item
    Set
    \(
      \tau \langle x_1,\dots,x_{n_t} \rangle
      \coloneq
      \tau_\ro.
    \)

    \item
    For every $i\in[n_t+1,m]$, set
    \(
      \tau \langle x_i \rangle\coloneq\pi'_i.
    \)

    \item
    Using $\phi$-realisability of $\cW$, take a $\cW$-completion of the
    resulting set of literals.
  \end{enumerate}

  Since $\cW$ is a satisfiability witness, $\tau$ admits a
  $(\cW,t)$-extension $\sigma$.
  Write
  \tightdisplay{\[
    \text{$\pi_i\coloneq \sigma\langle x_i \rangle$ for $i\in[m+1]$,}
    \qquad
    \pi'_{m+1}
    \coloneq
    \chc(\pi_{m+1};t).
  \]}%

  We next define an $(m+1)$-type $\sigma'$ in two steps:
  \begin{enumerate}[nosep]
    \item
    Set
    \(
      \sigma'\coloneq(\sigma\setminus\tau)\cup\tau'.
    \)

    \item
    Replace the $1$-type induced on $x_{m+1}$ by setting
    \(
      \sigma' \langle x_{m+1} \rangle
      \coloneq
      \pi'_{m+1}.
    \)
  \end{enumerate}

  We claim that the resulting $(m+1)$-type $\sigma'$ is a
  $(\cW',t)$-extension of $\tau'$.
  Clearly, $\sigma'$ extends $\tau'$, so it remains to verify that
  $\sigma'\models\psi_t$ and that $\TTTo(\sigma')\subseteq\cW'$.

  \begin{figure}[h]
    \centering
    \begin{tikzpicture}[scale=0.95]
  \node[draw, circle, minimum height=1.1cm] (node1) at (0,3) {$x_1$};
  \node[draw, circle, minimum height=1.1cm] at (2,3) {$x_2$};
  \node[node font=\Large] at (3.5,3) {$\dots$};
  \node[draw, circle, minimum height=1.1cm] at (5,3) {$x_{n_t}$};
  \node[draw, circle, minimum height=1.1cm] at (7,3) {$x_{n_t+1}$};
  \node[node font=\Large] at (8.5,3) {$\dots$};
  \node[draw, circle, minimum height=1.1cm] at (10,3) {$x_m$};
  \node[draw, circle, minimum height=1.1cm] at (12,3) {$x_{m+1}$};
  \draw[rounded corners=28.4pt] (-1,2) rectangle (6,4);

  \node at (0,1.5) {$\pi_1$};
  \node at (2,1.5) {$\pi_2$};
  \node at (5,1.5) {$\pi_{n_t}$};
  \node at (7,1.5) {$\pi_{n_t+1}$};
  \node at (10,1.5) {$\pi_m$};
  \node at (12,1.5) {$\pi_{m+1}$};

  \node at (0,-1.5) {$\pi'_1$};
  \node at (2,-1.5) {$\pi'_2$};
  \node at (5,-1.5) {$\pi'_{n_t}$};
  \node at (7,-1.5) {$\pi'_{n_t+1} = \pi_{n_t+1}$};
  \node at (10,-1.5) {$\pi'_m = \pi_m$};
  \node at (12,-1.5) {$\pi'_{m+1}$};

  \node[draw, circle, minimum height=1.1cm] (node2) at (0,-3) {$x_1$};
  \node[draw, circle, minimum height=1.1cm] at (2,-3) {$x_2$};
  \node[node font=\Large] at (3.5,-3) {$\dots$};
  \node[draw, circle, minimum height=1.1cm] at (5,-3) {$x_{n_t}$};
  \node[draw, circle, minimum height=1.1cm] at (7,-3) {$x_{n_t+1}$};
  \node[node font=\Large] at (8.5,-3) {$\dots$};
  \node[draw, circle, minimum height=1.1cm] at (10,-3) {$x_m$};
  \node[draw, circle, minimum height=1.1cm] at (12,-3) {$x_{m+1}$};
  \draw[rounded corners=28.4pt] (-1,-4) rectangle (6,-2);
  
  \draw[arrows=|-Triangle] (0,1) -- (0,-1);
  \node at (-0.5,0) {$\chc_{t_1}$};
  \draw[arrows=|-Triangle] (2,1) -- (2,-1);
  \node at (1.5,0) {$\chc_{t_2}$};
  \draw[arrows=|-Triangle] (5,1) -- (5,-1);
  \node at (4.5,0) {$\chc_{t_{n_t}}$};
  \draw[arrows=|-Triangle] (7,1) -- (7,-1);
  \node at (6.3,0) {$\chc_{t_{n_t}+1}$};
  \draw[arrows=|-Triangle] (10,1) -- (10,-1);
  \node at (9.5,0) {$\chc_{t_m}$};
  \draw[arrows=|-Triangle] (12,1) -- (12,-1);
  \node at (11.5,0) {$\chc_{t}$};

  \draw [decorate, decoration={brace,amplitude=5pt,raise=4ex}, very thick]
  (-1,4) -- (11,4) node[midway,yshift=3em]{$\tau$};
  \draw [decorate, decoration={brace,amplitude=5pt,raise=4ex}, very thick]
  (-1,5) -- (12.5,5) node[midway,yshift=3em]{$\sigma$};

  \draw [decorate, decoration={brace,amplitude=5pt,mirror,raise=4ex}, very thick]
  (-1,-4) -- (11,-4) node[midway,yshift=-3em]{$\tau'$};
  \draw [decorate, decoration={brace,amplitude=5pt,mirror,raise=4ex}, very thick]
  (-1,-5) -- (12.5,-5) node[midway,yshift=-3em]{$\sigma'$};

  \node at (-1.5,3) {$\tau_\ro$};
  \node at (-1.5,-3) {$\tau'_\ro$};
  \draw[arrows=|-Triangle] (-1.5,2.5) -- (-1.5,-2.5);
  \node at (-2,0) {$\chc_\ts$};
\end{tikzpicture}
    \caption{In the figure $\chc_\ts$ abbreviates $\tau_\ro \mapsto \chc(\tau_\ro;t_1,\dots,t_{n_t})$.}
    \label{fig:rysunek}
  \end{figure}
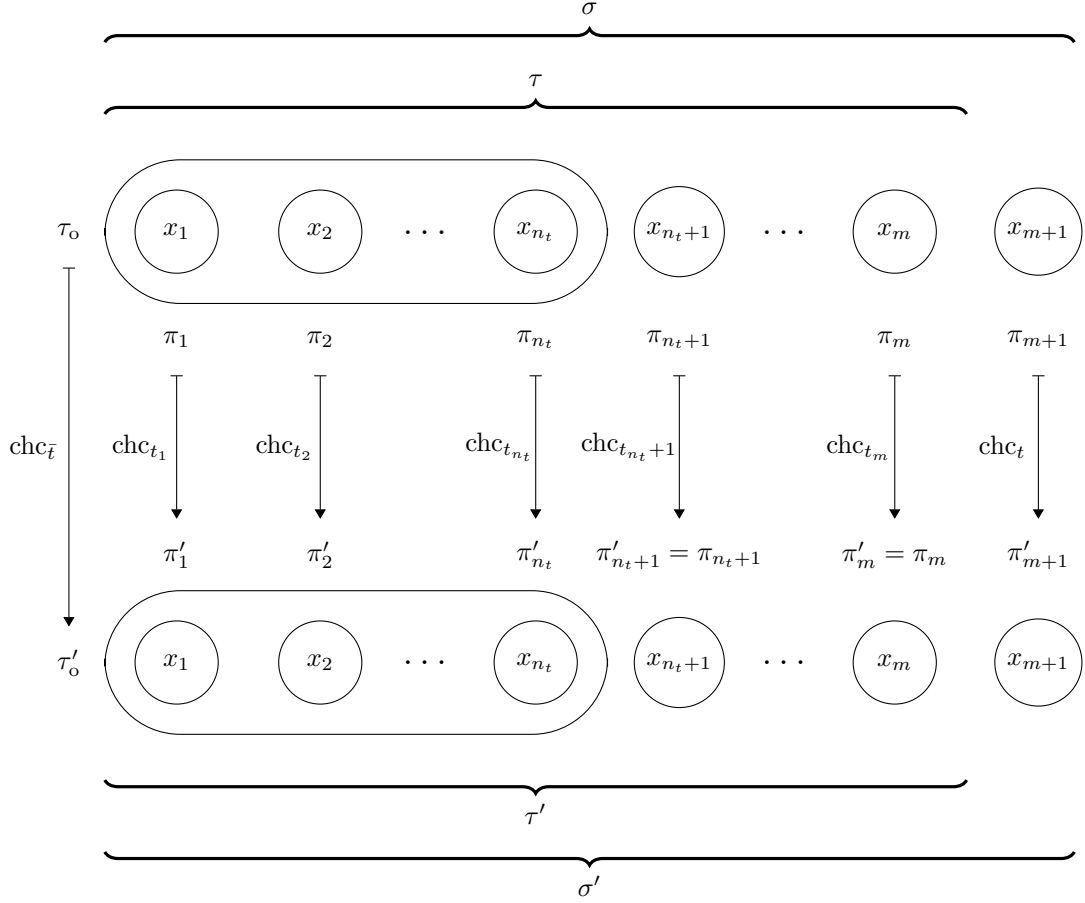

  Before verifying these two properties, we record the key properties
  of the construction:
  \tightdisplay{\[
    \text{$\pi_i\sim_{t_i}\pi'_i$ for $i\in[n_t]$,}
    \qquad
    \text{$\pi_i=\pi'_i$ for $i\in[n_t+1,m]$,}
    \qquad
    \pi_{m+1}\sim_t\pi'_{m+1},
  \]}%
  and
  \tightdisplay{\[
    \tau'_\ro=\chc(\tau_\ro;t_1,\dots,t_{n_t}).
  \]}%

  Figure~\ref{fig:rysunek} illustrates the above relationships among the types
  in the construction.
  The upper row represents $\tau$ and its extension $\sigma$, while the
  lower row represents the corresponding types $\tau'$ and $\sigma'$.
  The individual circles represent the $1$-types
  $\pi_1,\dots,\pi_{m+1}$ and
  $\pi'_1,\dots,\pi'_{m+1}$, whereas the large rounded
  regions represent the outer-types $\tau_\ro$ and $\tau'_\ro$ on the
  uniform variables $x_1,\dots,x_{n_t}$.
  The vertical arrows indicate the applications of $\chc$
  with the displayed parameters.
  Finally, the braces indicate which variables are covered by
  $\tau,\sigma,\tau'$, and $\sigma'$.

  We first verify that $\sigma'\models\psi_t$.
  Since $\sigma\models\psi_t$, it suffices to show that every atom of
  $\psi_t$ has the same truth value in $\sigma'$ as in $\sigma$.
  Recall that, by convention,
  \(
    n_t\in\{0\}\cup[2,m].
  \)
  We first consider the case $n_t\ge2$ and examine the possible forms of
  an atom $\alpha\in\Atoms(\psi_t)$.

  \begin{enumerate}[nosep,label={(\alph*)}]
    \item\label{witness-small-1-types:item1}
    Suppose that $\alpha\in\AAA_t$ and
    $\Vars(\alpha)\subseteq\{x_i\}$ for some $i\in[n_t]$.
    Since $n_t\ge2$, we have $\AAA_t\subseteq\AAA^*$.
    Hence, by the definition of $\sim_{t_i}$ and
    $\pi_i\sim_{t_i}\pi'_i$,
    we have
    \(
      \sigma\models\alpha
      \Longleftrightarrow
      \sigma'\models\alpha.
    \)

    \item\label{witness-small-1-types:item2}
    Suppose that $\alpha\in\AAA_t$ and
    $\Vars(\alpha)\subseteq\{x_i\}$ for some
    $i\in[n_t+1,m]$.
    Since $\pi_i=\pi'_i$, we directly obtain 
    \(
      \sigma\models\alpha
      \Longleftrightarrow
      \sigma'\models\alpha.
    \)

    \item\label{witness-small-1-types:item3}
    Suppose that $\alpha\in\BBB_t$.
    Then $\alpha$ involves precisely the uniform variables
    $x_1,\dots,x_{n_t}$, and hence its truth value in $\sigma$ is
    determined by the induced outer-type
    \(
      \tau_\ro
      =
      \sigma \langle x_1,\dots,x_{n_t} \rangle,
    \)
    whereas its truth value in $\sigma'$ is determined by
    \(
      \tau'_\ro
      =
      \sigma' \langle x_1,\dots,x_{n_t} \rangle.
    \)
    By construction,
    \(
      \tau'_\ro
      =
      \chc(\tau_\ro;t_1,\dots,t_{n_t}).
    \)
    Since the mapping
    $\tau_\ro\mapsto \chc(\tau_\ro;t_1,\dots,t_{n_t})$
    changes only the induced $1$-types, it preserves all literals involving
    at least two distinct variables.
    As $\alpha$ involves precisely $n_t\ge2$ distinct variables, we obtain
    \(
      \sigma\models\alpha
      \Longleftrightarrow
      \sigma'\models\alpha.
    \)

    \item\label{witness-small-1-types:item4}
    Suppose that $\alpha\in\CCC'_t$.
    Then $\alpha$ is determined by the $1$-type induced on
    $x_{m+1}$.
    Since $\pi_{m+1}\sim_t\pi'_{m+1}$, the definition of $\sim_t$
    yields
    \(
      \sigma\models\alpha
      \Longleftrightarrow
      \sigma'\models\alpha.
    \)

    \item\label{witness-small-1-types:item5}
    Finally, suppose that
    $\alpha\in\CCC_t\setminus\CCC'_t$.
    Then $\alpha$ contains $x_{m+1}$ and at least one further
    variable.
    Its corresponding literal $\pm \alpha$ is 
    contained in $\sigma \setminus \tau$,
    but does not belong to the $1$-type induced on $x_{m+1}$.
    It is therefore
    retained in $\sigma'$.
    Hence,
    \(
      \sigma\models\alpha
      \Longleftrightarrow
      \sigma'\models\alpha.
    \)
  \end{enumerate}

  Since the cases above collectively cover all atoms of $\psi_t$,
  it follows that $\sigma'\models\psi_t$.

  We now consider the case $n_t=0$.
  The argument above applies with only minor simplifications.
  In this case, $\tau_\ro$ and $\tau'_\ro$ are simply the $0$-type of
  $\cW$, so no choice of $1$-types is involved in their construction.
  Moreover, cases~\ref{witness-small-1-types:item1}
  and~\ref{witness-small-1-types:item3} cannot occur, while all remaining
  cases apply verbatim.
  Thus, $\sigma'\models\psi_t$ also when $n_t=0$.

  We now verify that $\TTTo(\sigma')\subseteq\cW'$.
  Since $\sigma'$ extends $\tau'$ and $\tau'$ is a $\cW'$-requirement,
  every outer-type induced by $\sigma'$ on any sequence of the first $m$
  variables already belongs to $\cW'$.
  It remains therefore to verify outer-types induced on sequences involving $x_{m+1}$.
  Hence, consider
  such a sequence of pairwise distinct variables
  \(
    x_{i_1},\dots,x_{i_k},x_{m+1}.
  \)
  Let
  \(
    \tilde{\tau}_\ro
    \coloneq
    \sigma \langle x_{i_1},\dots,x_{i_k},x_{m+1} \rangle.
  \)
  Since $\sigma$ is a $(\cW,t)$-extension, we have
  $\tilde{\tau}_\ro\in\cW$.
  By construction of $\sigma'$,
  \tightdisplay{\[
    \sigma' \langle x_{i_1},\dots,x_{i_k},x_{m+1} \rangle
    =
    \chc(\tilde{\tau}_\ro;t_{i_1},\dots,t_{i_k},t).
  \]}%
  Hence, by the construction of $\cW'$,
  \(
    \sigma' \langle x_{i_1},\dots,x_{i_k},x_{m+1} \rangle
    \in\cW'.
  \)
  Thus,
  \(
    \TTTo(\sigma')\subseteq\cW'.
  \)

  We have shown that $\sigma'$ extends $\tau'$, satisfies $\psi_t$, and
  induces only outer-types from $\cW'$.
  Therefore, $\sigma'$ is a $(\cW',t)$-extension of $\tau'$, as claimed.
\end{proof}

\paragraph{Reduction of witness size.}
We now reduce the size of $\cW'$ by introducing another family of equivalence relations, this time on the $\cW'$-requirements.

For $t \in \cI$ and $\cW'$-requirements $\tau_1,\tau_2$, we define $\tau_1 \approx_t \tau_2$ if
\begin{itemize}[nosep]
  \item $\tau_1 \models \alpha \iff \tau_2 \models \alpha$ for every $\alpha \in \AAA_t \cup \BBB_t$;
  \item $\tau_1 \langle x_i \rangle = \tau_2 \langle x_i \rangle$ for every $i \in [m]$.
\end{itemize}

Let $N_{\approx_t}$ denote the number of equivalence classes of $\approx_t$.
There are $|\cW' \cap \TTT^1(\Sigma)|^m$ sequences of $1$-types of length $m$
and at most $2^{s_\alpha+s_\beta}$ truth assignments to the atoms in $\AAA_t \cup \BBB_t$. Hence,
\tightdisplay{\begin{equation}\label{eqn:size-estimate-2}
  \sum\nolimits_{t \in \cI} N_{\approx_t} \;\le\; |\cI| \cdot |\cW' \cap \TTT^1(\Sigma)|^m \cdot 2^{s_\alpha+s_\beta}.
\end{equation}}%

We construct the final set $\cW''$ of outer-types in two stages:
\begin{enumerate}[nosep]
  \item\label{construction-Wbis-item1}
  For each $t\in\cI$ and each equivalence class $\cC$ of $\approx_t$,
  choose a representative $\tau\in\cC$ and a $(\cW',t)$-extension
  $\sigma$ of $\tau$, and include all outer-types from $\TTTo(\sigma)$.

  \item\label{construction-Wbis-item2}
  For each $k\in[m+1]$ and each sequence of $1$-types
  $\pi_1,\dots,\pi_k\in\cW'\cap\TTT^1(\Sigma)$,
  choose a single $k$-outer-type $\tau_\ro\in\cW'$ such that
  \(
    \tau_\ro \langle x_i \rangle =\pi_i
  \)
  for every $i\in[k]$,
  and include all $k$-outer-types obtainable from $\tau_\ro$ by permuting
  variables.
\end{enumerate}

\begin{claim}
  The set $\cW''$ is a satisfiability witness for $\phi$.
\end{claim}

\begin{proof}
  We first verify the $\phi$-extension property.
  Let $\tau$ be any $\cW''$-requirement and let $t\in\cI$.
  Since $\cW''\subseteq\cW'$, the type $\tau$ is also a
  $\cW'$-requirement.
  Let $\tau'$ be the representative of the $\approx_t$-class of $\tau$
  selected in stage~\ref{construction-Wbis-item1}, and let $\sigma'$ be
  the selected $(\cW',t)$-extension of $\tau'$.
  By Claim~\ref{claim:type-surgery},
  \(
    \sigma \coloneq (\sigma'\setminus\tau')\cup\tau
  \)
  extends $\tau$ and satisfies $\psi_t$.
  Moreover, every outer-type induced by $\sigma$ is induced either by
  $\sigma'$ or by $\tau$.
  The former belong to $\cW''$ by
  stage~\ref{construction-Wbis-item1}, while the latter belong to
  $\cW''$ because $\tau$ is a $\cW''$-requirement.
  Hence, $\sigma$ is in fact a $(\cW'',t)$-extension of $\tau$.

  It remains to verify that $\cW''$ is $\phi$-realisable.
  The outer-types included in stage~\ref{construction-Wbis-item1} already
  satisfy all conditions of $\phi$-realisability except, potentially, the
  condition requiring every sequence of $1$-types of length at most $m+1$
  to be induced by some outer-type.
  The outer-types included in stage~\ref{construction-Wbis-item2} ensure
  that this remaining condition is satisfied.
\end{proof}

\paragraph{Size bound.}
An $(m{+}1)$-type can induce at most
\(
  \sum_{i=0}^{m+1}\binom{m+1}{i}\cdot i!
  \le
  (m+1)^{m+2}
\)
distinct outer-types.
For each equivalence class $\cC$ of $\approx_t$,
stage~\ref{construction-Wbis-item1} includes all outer-types induced by
one selected $(m+1)$-type extending the representative of $\cC$.
Moreover, stage~\ref{construction-Wbis-item2} includes at most
$(m+1)!$ additional outer-types for every sequence of $1$-types of
length at most $m+1$.
Therefore,
\tightdisplay{\begin{equation}\label{eqn:size-estimate-3}
  |\cW''|
  \;\le\;
  \sum\nolimits_{t\in\cI}
    N_{\approx_t}\cdot(m+1)^{m+2}
  \;+\;
  \sum\nolimits_{k\in[m+1]}
    |\cW'\cap\TTT^1(\Sigma)|^k\cdot(m+1)!.
\end{equation}}%

Combining
inequalities~\eqref{eqn:size-estimate-1},~\eqref{eqn:size-estimate-2},
and~\eqref{eqn:size-estimate-3} yields
\tightdisplay{\begin{equation}\label{eqn:size-estimate-4}
  |\cW''|
  \;\le\;
  |\cI|^{m+1}
  \cdot
  2^{m\cdot(s^*_\alpha+s'_\gamma)+s_\alpha+s_\beta}
  \cdot
  (m+1)^{m+2}
  \;+\;
  (m+1)\cdot
  |\cI|^{m+1}
  \cdot
  2^{(m+1)\cdot(s^*_\alpha+s'_\gamma)}
  \cdot
  (m+1)!.
\end{equation}}%

By Lemma~\ref{lemma:reduction-to-simple-form},
$s^*_\alpha=\cO(|\phi_{\rm in}|^2)$,
$|\cI|=2^{\cO(|\phi_{\rm in}|\log|\phi_{\rm in}|)}$, and
$m$, $s_\alpha$, $s_\beta$, and $s'_\gamma$ are all
$\cO(|\phi_{\rm in}|)$.
Substituting these bounds into~\eqref{eqn:size-estimate-4} gives
\(
  |\cW''|=2^{\cO(|\phi_{\rm in}|^3)}.
\)

We have therefore established an exponential bound on the size of minimal satisfiability witnesses.
By applying Lemma~\ref{lemma:random-models} to such a witness, we obtain that every satisfiable sentence in $\DMDK{}$ admits a model of exponential size.
Hence, $\DMDK{}$ has the exponential model property and its satisfiability problem is in $\NExpTime$.

For the lower bound, recall that $\DMDK{}$ subsumes the Scott normal form of equality-free $\FOt{}$, whose satisfiability
problem is $\NExpTime$-hard; see, e.g.,~\cite{PH23}.

\begin{theorem}\label{theorem:main-DMDK}
  The class \DMDK{} has the exponential model property.
  The satisfiability problem for \DMDK{} is \NExpTime-complete.
\end{theorem}

\section{The Decidability of \DMDK{} with Equivalence}
\label{sec:equivalence}

In this section, we study extensions of $\DMDK{}$ with distinguished
equivalence and transitive relations.
Our main result is that satisfiability for $\DMDK{}$ extended with a
single equivalence relation is $\TwoNExpTime$-complete.
We also show that allowing two equivalence relations, or a single
transitive relation, leads to undecidability.
We begin by introducing the relevant notation.

\paragraph{Notation.}
For a fragment $\cF$, we write $\cF + \EEE$ for its extension with a
single distinguished binary relation symbol $\EEE$, called the
\emph{equivalence symbol}.
Models are restricted to structures in which $\EEE$ is interpreted as
an equivalence relation.
More generally, for $k \in \N$, we write $\cF + k\EEE$ for the
extension with $k$ distinguished equivalence symbols
$\EEE_1,\dots,\EEE_k$.

We write $\cF+\TR$ for the extension with a single distinguished
binary relation symbol $\TR$, called the \emph{transitive symbol},
whose interpretation is required to be transitive.
Analogously, we write $\cF+\TRrefl$ for the extension with a single
distinguished binary relation symbol $\TRrefl$, called the
\emph{reflexive-transitive symbol}, whose interpretation is required
to be both reflexive and transitive.

\paragraph{Free use.}
In $\DMK{} + \EEE$, the equivalence symbol $\EEE$ may, without loss of
generality, be used \emph{freely}, meaning that atoms involving $\EEE$
need not satisfy the restrictions of $\DMK{}$ on $\phi$-prefixes.
Indeed, assuming that the sentence is in negation normal form, which we
may always do, every literal involving $\EEE$ can be rewritten using
the equivalence
\tightdisplay{\begin{equation}\label{eqn:EEE-sugar}
  \pm\, x \EEE y
  \quad\Longleftrightarrow\quad
  \exists z\,
  (x \EEE z \wedge \pm\, z \EEE y),
\end{equation}}%
where $z$ is fresh.
Since $\EEE$ is interpreted as an equivalence relation, the rewriting
is sound, and both introduced occurrences of $\EEE$ have
$\phi$-prefixes ending with an existential quantifier.

The same observation applies to $\DMK{}+\TRrefl$: the distinguished
symbol $\TRrefl$ may, without loss of generality, be used freely.
Indeed, the above rewriting applies to any reflexive and transitive relation.
This is the reason why we consider a distinguished reflexive-transitive relation rather than merely a transitive one.
For the undecidability results established below, however, this entails
no loss of generality: reflexivity can itself be expressed in $\DMK{}$
by the sentence
\(
  \forall x\,\TR(x,x).
\)

\subsection{Lower Bounds and Undecidability}
\label{subsec:EEE-lower}

We now establish the $\TwoNExpTime{}$ lower bound for $\DMK{}+\EEE$
and show that the closely related extensions $\DMK{}+2\EEE$ and
$\DMK{}+\TRrefl$ are already undecidable.

\paragraph{Lower bounds.}
Kiero\'nski and Kuusisto~\cite{KK15} studied the uniform
one-dimensional fragment $\UF$ with free use of the equivalence symbol
$\EEE$; under our convention, this logic is denoted by $\UF+\EEE$.
They proved that satisfiability for $\UF+\EEE$ is
$\TwoNExpTime{}$-complete, with the lower bound already holding for the
equality-free fragment.
A direct inspection of their construction
in~\cite[Section~5.1]{KK15} shows that, up to straightforward logical
manipulations, all formulas used there belong to $\DMK{}+\EEE$.
In particular,
free occurrences of $\EEE$ are permitted by~\eqref{eqn:EEE-sugar}.
Moreover, since all relation symbols involved in the construction are unary or binary, the resulting conjuncts can all be
taken to have uniformity~$2$ and hence merged into a single
$\DMK{}+\EEE$ sentence.
Thus, satisfiability for $\DMK{}+\EEE$ is
$\TwoNExpTime{}$-hard.

We strengthen this lower bound by restricting the non-distinguished part of the signature
to unary predicates.
Let $\MON{}$ denote the monadic fragment of equality-free, constant-free $\FO{}$,
that is, the class of all equality-free sentences over signatures containing only
unary predicate symbols.
We consider its extension $\MON{}+\EEE$, in which the distinguished
equivalence symbol $\EEE$ is the only binary symbol available and can be used freely.
By reducing from the doubly exponential grid tiling problem, we show that the
$\TwoNExpTime{}$ lower bound holds already for $\MON{}+\EEE$.
The details of this reduction are deferred to Appendix~\ref{appendix:equivalences}.

\begin{proposition}\label{prop:EEE-lower}
  The satisfiability problem for $\MON{}+\EEE$ is
  $\TwoNExpTime{}$-hard.
\end{proposition}

Since $\MON{}\subseteq\DMK{}$ and free use of $\EEE$ is permitted
by~\eqref{eqn:EEE-sugar}, we have
$\MON{}+\EEE\subseteq\DMK{}+\EEE$.
Thus, Proposition~\ref{prop:EEE-lower} in particular yields the
$\TwoNExpTime{}$ lower bound for $\DMK{}+\EEE$.

\paragraph{Undecidability.}
In the same work, Kiero\'nski and Kuusisto~\cite[Section~5.5]{KK15}
proved that satisfiability for $\UF$ with either
free use of two equivalence
symbols or free use of a single transitive symbol is undecidable.
Under our convention, these logics are denoted by $\UF+2\EEE$ and
$\UF+\TR$, respectively.
Inspection of their constructions shows that all non-distinguished
relation symbols used in the reductions are in fact unary.
Consequently, their results already yield undecidability of
$\MON{}+2\EEE$ and $\MON{}+\TR$.

Since free use of the equivalence symbols is permitted
by~\eqref{eqn:EEE-sugar}, it follows directly that satisfiability for
$\DMK{}+2\EEE$ is undecidable.
In contrast, the rewriting in~\eqref{eqn:EEE-sugar} cannot be applied
to a merely transitive relation, since both reflexivity and transitivity
are required.
We therefore modify the construction for $\MON{}+\TR$ by adding all
self-loops to the distinguished relation.
This makes the relation reflexive while preserving transitivity and,
by inspection of the construction, does not affect the correctness of
the reduction.
Hence, satisfiability is undecidable already for
$\MON{}+\TRrefl$ and, by~\eqref{eqn:EEE-sugar}, also for
$\DMK{}+\TRrefl$.
We provide further details in Appendix~\ref{appendix:equivalences}.

\begin{proposition}\label{prop:EEE-undec}
  The satisfiability problems for $\MON{}+2\EEE$ and
  $\MON{}+\TRrefl$ are co-$RE$-complete.
\end{proposition}

\subsection{An Upper Bound}
\label{subsec:EEE-upper}

We now turn to the main result of this section: the decidability of
$\DMDK{} + \EEE$.
More precisely, we establish a $\TwoNExpTime{}$ upper bound, matching
the lower bound from Subsection~\ref{subsec:EEE-lower}.

The Skolemisation procedure of Proposition~\ref{prop:reduction} remains
sound for $\DMDK{}+\EEE$.
Indeed, the input sentence and its Skolem normal form are satisfied in
precisely the same structures, up to the interpretation of the fresh
predicates introduced during the reduction.
In particular, the intended interpretation of the equivalence symbol
$\EEE$ is preserved.
Hence, we work with $\SNF{}+\EEE$ rather than the full class
$\DMDK{}+\EEE$.

The two preliminary reductions from
Subsection~\ref{subsec:prelim-reductions} extend to $\SNF{}+\EEE$ as
well: we may adopt the standard name assumption and apply
Lemma~\ref{lemma:reduction-to-simple-form} to reduce to simple
assignments.

Accordingly, from this point on we adopt the standard name assumption,
fix an input sentence $\phi_{\rm in}\in\SNF{}+\EEE$, and let $\phi$ be
the sentence obtained from $\phi_{\rm in}$ by
Lemma~\ref{lemma:reduction-to-simple-form}.
As before, $\phi_{\rm in}$ is retained only for measuring the input
size, while all subsequent arguments concern $\phi$.
We use for $\phi$ the notation introduced in
Definition~\ref{def:skolem-class}.

Let $\Sigma_0 = \langle \Rels, \Cons \rangle$ be the \emph{base
signature} of $\phi$, consisting of the free relation symbols and
constants, and let $\Sigma_\EEE = \langle \Rels \cup \{\EEE\}, \Cons \rangle$
be the full signature of $\phi$, obtained by extending $\Sigma_0$ with
the equivalence symbol $\EEE$.

\paragraph{Class labels.}
Our first step is to encode the equivalence classes of $\EEE$ by unary
predicates.
For each $j \in \N$, introduce a fresh unary predicate $\LLL_j$, called
a \emph{class label}, and let
\(
  \LLL \coloneq \{\LLL_j \mid j \in \N\}.
\)
We use the auxiliary signature
\( \Sigma_\LLL \coloneq \langle \Rels \cup \LLL, \Cons \rangle, \)
obtained from the base signature $\Sigma_0$ by adding the class-label predicates.
Note that, by construction, the signature $\Sigma_\LLL$ is countably infinite.

Given a $\Sigma_\EEE$-structure $\fA$ in which $\EEE^\fA$ is an equivalence
relation, we associate with it a $\Sigma_\LLL$-structure by assigning a
distinct class label to each equivalence class.
More precisely, enumerate the classes in $A/\EEE^\fA$ and assign to
each $\cC \in A/\EEE^\fA$ one previously unused class label $\LLL_j$ by interpreting \( \LLL_j^\fA = \cC. \)
If $A/\EEE^\fA$ is finite, all unused class labels $\LLL_j$ are
interpreted as empty.

Thus, now every element carries exactly one class label, and two elements
are $\EEE$-equivalent precisely when they carry the same label, that is,
for every elements $a,b \in A$,
\tightdisplay{\[
  \fA \models
  a \EEE b
  \quad\Longleftrightarrow\quad
  \fA \models \LLL_j(a) \wedge \LLL_j(b)
  \text{ for some $j \in \N$.}
\]}%

We then forget the interpretation $\EEE^{\fA}$ and view $\fA$
as a $\Sigma_\LLL$-structure.
No information is lost in doing so, since $\EEE^{\fA}$ can be recovered entirely from the
class labels using the above equivalence.
Conversely, any $\Sigma_\LLL$-structure in which every element carries exactly one
class label induces in this way an equivalence relation, which we take as the interpretation of $\EEE$.
Accordingly, from this point on we work exclusively with the signature $\Sigma_\LLL$ and regard $\EEE$ as implicitly encoded.
Under this convention, $\phi$ can still be evaluated in
$\Sigma_\LLL$-structures.
Importantly, $\phi$ has no direct access to the predicates $\LLL_j$ and
can only use the derived relation $\EEE$.

\paragraph{Types and outer-types.}
We henceforth work with types and outer-types over the auxiliary
signature $\Sigma_\LLL$.
They are defined as before, with the additional requirement that every
variable or constant occurring in a type or outer-type satisfies exactly one class-label predicate.
Thus, despite the infinite signature $\Sigma_\LLL$, individual types
and outer-types admit finite descriptions:
each is represented by its $\Sigma_0$-reduct together with the class labels of its variables and constants.

As with $\Sigma_\LLL$-structures, we adopt the convention that in a
type or outer-type the equivalence relation $\EEE$ is induced
by the class-label predicates $\LLL_j$.

The following terminology will be convenient.
A class label $\LLL_j$ \emph{occurs} in a type or outer-type $\tau$ if
$\tau$ contains a positive $\LLL_j$-literal involving a variable or a
constant; the corresponding variable or constant then \emph{carries}
the class label $\LLL_j$.
A $1$-type $\pi$ \emph{carries} a class label $\LLL_j$ if
\(
  \pi\models\LLL_j(x_1).
\)
More generally, $\LLL_j$ \emph{occurs} in a set of outer-types if it
occurs in one of its members.

\paragraph{Satisfiability witnesses.}
We adapt Definitions~\ref{definition:phi-realisability}
and~\ref{definition:phi-extension}, and consequently
Definition~\ref{def:satisfiability-witness}, to the present setting by taking all types and
outer-types occurring in those definitions over $\Sigma_\LLL$.
Apart from this change of signature, the definitions remain unchanged.
In particular, satisfiability witnesses may now be infinite.

The notion of a satisfiability witness remains complete,
that is, whenever $\phi$ is satisfiable, it admits a satisfiability witness.
This is established by the argument from Subsection~\ref{subsec:criterion}, which applies
unchanged to the class-label encoding.

It remains to establish the converse direction with suitable size
bounds.
We show that whenever $\phi$ admits a satisfiability witness, it admits
a finite satisfiability witness of size doubly exponential in
$|\phi_{\rm in}|$.
We then apply the random model construction of
Subsection~\ref{subsec:criterion} to obtain a finite model of doubly
exponential size.
This yields the desired $\TwoNExpTime{}$ upper bound.

\paragraph{Local reduction of $1$-types.}
Fix a satisfiability witness $\cW$ for $\phi$.
In $\cW$, the number of $1$-types carrying any fixed class label $\LLL_j$ may
be doubly exponential.
Our first step is to reduce this number to singly exponential.
This reduction is crucial for obtaining the optimal $\TwoNExpTime$
upper bound; without it, the resulting complexity bound would be only $\ThreeNExpTime$.

We adapt the strategy from
Subsection~\ref{sec:exponential-certificates}.
Recall that there, for every $t \in \cI$, we introduced an equivalence
relation $\sim_t$ on $1$-types.
In the present setting, we refine $\sim_t$ by additionally requiring
equivalent $1$-types to carry the same class label.

Let the sets $\AAA^*$ and $\CCC'_t$ be as in~\eqref{eqn:AAA-star},
and let $[x_1/\,\cdot\,]$ denote the substitution identifying every variable with $x_1$.
For $t \in \cI$ and $1$-types $\pi_1,\pi_2 \in \cW \cap \TTT^1(\Sigma_\LLL)$,
we set $\pi_1 \sim_t \pi_2$ if
\begin{enumerate}[nosep]
  \item $\pi_1 \models \alpha[x_1/\,\cdot\,] \Longleftrightarrow \pi_2 \models \alpha[x_1/\,\cdot\,]$ for every $\alpha \in \AAA^* \cup \CCC'_t$;
  \item $\pi_1 \models \LLL_j(x_1)$ and $\pi_2 \models \LLL_j(x_1)$ for some $j \in \N$.
\end{enumerate}

For every $t \in \cI$, fix a choice function
\(
  \chc_t \colon
  (\cW \cap \TTT^1(\Sigma_\LLL))/_{\sim_t}
  \longrightarrow
  \cW \cap \TTT^1(\Sigma_\LLL)
\)
selecting representatives of $\sim_t$-classes.
Recall from Subsection~\ref{sec:exponential-certificates} that,
for a $k$-outer-type $\tau_\ro \in \cW$ and $t_1,\dots,t_k \in \cI$,
we write $\chc(\tau_\ro;t_1,\dots,t_k)$ for the $k$-outer-type obtained
from $\tau_\ro$ by replacing each induced $1$-type $\pi_i \coloneq \tau_\ro \langle x_i \rangle$
with $\chc_{t_i}([\pi_i]_{\sim_{t_i}})$,
while leaving all remaining literals unchanged.

We now define
\tightdisplay{\[
  \cW' \;\coloneq\;
  \bigl\{ \chc(\tau_\ro;t_1,\dots,t_k) \mid k \in [0,m+1],\, \tau_\ro \in \cW \cap \TTT^k_\ro(\Sigma_\LLL),\, t_1,\dots,t_k \in \cI \bigr\}. 
\]}%

The definition of $\cW'$ differs from that used in
Subsection~\ref{sec:exponential-certificates} only in that $1$-type
replacements are additionally required to preserve class labels.
The argument of Claim~\ref{claim:witness-small-1-types} extends
naturally to account for the additional predicates $\LLL_j$.
Thus, $\cW'$ is a satisfiability witness for $\phi$.

It remains to bound the number of $1$-types carrying a fixed label.
Let $s^*_\alpha \coloneq |\AAA^*|$ and $s'_\gamma \coloneq \max_{t \in \cI} |\CCC'_t|$.
Fix $j \in \N$.
For a given $t \in \cI$, there are at most
$2^{s^*_\alpha+s'_\gamma}$
distinct $\sim_t$-classes carrying the label $\LLL_j$.
As $\cW'$ retains one representative for each $\sim_t$-class for
each $t$, we obtain
\tightdisplay{\begin{equation}\label{eqn:E-first-witness-bound}
  \big|
    \big\{
      \pi \in \cW' \cap \TTT^1(\Sigma_\LLL)
      \mid
      \pi \models \LLL_j(x_1)
    \big\}
  \big|
  \le
  |\cI| \cdot 2^{s^*_\alpha+s'_\gamma}.
\end{equation}}%

By Lemma~\ref{lemma:reduction-to-simple-form},
the right-hand side of~\eqref{eqn:E-first-witness-bound} is
$2^{\cO(|\phi_{\rm in}|^2)}$.
Therefore, the number of $1$-types in $\cW'$ carrying any fixed class
label is singly exponential in $|\phi_{\rm in}|$, as desired.
At this stage, however, $\cW'$ may still be infinite, since infinitely
many distinct class labels may occur.

\paragraph{Profiles of class labels.}
Our next goal is therefore to transform the satisfiability witness
$\cW'$ into one with only finitely many class labels.
The key observation is that class labels carrying the same collection
of $1$-types, up to the choice of the class label itself, should be
interchangeable.
To make this idea precise, we associate with each class label a
\emph{profile} recording this collection of $1$-types and, additionally,
the constants carrying that label.
We then close $\cW'$ under interchanges of labels having the same profile.
This will allow us to discard all but a bounded number of
labels of each profile: whenever a discarded class label is needed to provide
an existential witness, it can be replaced by one of the retained labels
of the same profile.

For every $j \in \N$,
let
\tightdisplay{\[
  R_j
  \;\coloneq\;
  \big\{
    \pi \in \TTT^1(\Sigma_0)
    \mid
    \text{$\pi$ is the $\Sigma_0$-reduct of a $1$-type in $\cW'$ that contains $\LLL_j(x_1)$}
  \big\}
\]}%
and
\tightdisplay{\[
  C_j
  \;\coloneq\;
  \big\{c \in \Cons \mid \text{the unique $0$-type of $\cW'$ contains $\LLL_j(c)$} \big\}.
\]}%
We call the pair $P_j \coloneq (R_j,C_j)$ the \emph{profile} of the class label $\LLL_j$.

The set $R_j$ records the underlying $1$-types over the base signature
$\Sigma_0$ that carry the class label $\LLL_j$, while $C_j$ records the
constants carrying that label.
Hence, the number of distinct profiles $P_j$ is clearly finite.
More precisely, it is at most doubly exponential in
$|\phi_{\rm in}|$.

Indeed, by~\eqref{eqn:E-first-witness-bound}, the number of distinct
sets $R_j$ is bounded by
\tightdisplay{\begin{equation}\label{eqn:E-second-witness-bound}
  \bigl|\{R_j\mid j\in\N\}\bigr|
  \;\le\;
  \bigl|\TTT^1(\Sigma_0)\bigr|^{2^{\cO(|\phi_{\rm in}|^2)}}
  \;\le\;
  \Bigl(
    2^{2^{\cO(|\phi_{\rm in}|\log|\phi_{\rm in}|)}}
  \Bigr)^{2^{\cO(|\phi_{\rm in}|^2)}}
  \;\le\;
  2^{2^{\cO(|\phi_{\rm in}|^2)}}.
\end{equation}}%
The second component $C_j$ contributes only linearly many additional
possibilities: if $C_j\neq\emptyset$, then no other class label has the
same profile as $\LLL_j$.

\paragraph{A highly symmetric witness.}
We now close $\cW'$ under permutations of class labels having the same
profile.
We say that a bijection $\rho\colon\N\to\N$ is
\emph{profile-preserving} if
\tightdisplay{\[
  P_j=P_{\rho(j)}
  \qquad\text{for every }j\in\N.
\]}%
A profile-preserving bijection $\rho$ acts on types and outer-types by
renaming class labels: each literal $\pm\LLL_j(a)$ is replaced by
$\pm\LLL_{\rho(j)}(a)$, while all other literals remain unchanged.
We then set
\tightdisplay{\[
  \cW''
  \;\coloneq\;
  \bigl\{
    \rho(\tau)
    \mid
    \tau\in\cW'
    \text{ and $\rho$ is a profile-preserving bijection}
  \bigr\}.
\]}%

We claim that $\cW''$ is a satisfiability witness.
First, observe that for every $1$-type
$\pi\in\TTT^1(\Sigma_\LLL)$ and every profile-preserving bijection
$\rho$,
we have
\(
  \pi\in\cW'
  \iff
  \pi\in\rho(\cW')
\).
It follows that the construction of $\cW''$ preserves the $1$-types:
\(
  \pi \in \cW''
  \iff
  \pi \in \cW'.
\)
Consequently,
\tightdisplay{\begin{equation}\label{eqn:rho-equivalence}
  \pi\in\cW''
  \iff
  \pi\in\cW'
  \iff
  \pi\in\rho(\cW').
\end{equation}}%

It follows directly that $\cW''$ is $\phi$-realisable.
Indeed, let $\pi_1,\dots,\pi_k\in\cW''$ be $1$-types, where
$k\le m+1$.
By~\eqref{eqn:rho-equivalence},
$\pi_1,\dots,\pi_k\in\cW'$.
Since $\cW'$ is $\phi$-realisable, there exists a $k$-outer-type
$\tau_\ro\in\cW'$ inducing $\pi_1,\dots,\pi_k$.
By the identity bijection, we have $\cW'\subseteq\cW''$, and hence $\tau_\ro\in\cW''$.
The remaining conditions of $\phi$-realisability are straightforward.

We next verify the $\phi$-extension property.
Fix $t\in\cI$ and a
$\cW''$-requirement $\tau$.
Write
\tightdisplay{\[
  \tau_\ro\coloneq \tau \langle x_1,\dots,x_{n_t} \rangle,\qquad
  \pi_i \coloneq \tau\langle x_i\rangle\text{ for }i \in [n_t+1,m].
\]}%

Thus, $\tau_\ro$ is the outer-type induced on the uniform variables $x_1,\dots,x_{n_t}$,
while $\pi_{n_t+1},\dots,\pi_m$ are the $1$-types induced on the remaining universal variables.

Since $\tau$ is a $\cW''$-requirement, we have
$\tau_\ro\in\cW''$.
Choose a profile-preserving bijection $\rho$ witnessing this membership,
that is, such that $\tau_\ro\in\rho(\cW')$.
By~\eqref{eqn:rho-equivalence}, each $1$-type
$\pi_i$ is in $\rho(\cW')$.

We now construct a $\cW'$-requirement $\tau'$ in three steps:
\begin{enumerate}[nosep]
  \item\label{witness-bijections-step1}
  Set
  \(
    \tau'\langle x_1,\dots,x_{n_t} \rangle
    \coloneq
    \rho^{-1}(\tau_\ro).
  \)

  \item\label{witness-bijections-step2}
  For every $i\in[n_t+1,m]$, set
  \(
    \tau'\langle x_i \rangle
    \coloneq
    \rho^{-1}\bigl(\pi_i\bigr).
  \) 

  \item\label{witness-bijections-step3}
  Using $\phi$-realisability of $\cW'$, take a $\cW'$-completion of the resulting set of literals.
\end{enumerate}

Recall from Subsection~\ref{sec:exponential-certificates} that
a $\cW'$-completion of a set of literals $\Gamma$ is any $m$-type $\tau'$ such that
\(\Gamma\subseteq\tau'\)
and
\(\TTTo(\tau')\subseteq\cW'.\)
Therefore, the constructed $m$-type $\tau'$ is indeed a $\cW'$-requirement.

Let $\sigma'$ be a $(\cW',t)$-extension of $\tau'$, and set
\(
  \sigma
  \coloneq
  \bigl(\rho(\sigma')\setminus\rho(\tau')\bigr)\cup\tau.
\)
By steps~\ref{witness-bijections-step1} and~\ref{witness-bijections-step2},
the $m$-types $\tau$ and $\rho(\tau')$
agree on all atoms in $\AAA_t\cup\BBB_t$
and induce the same sequences of $1$-types.
In particular, they agree on the class-label literals.
Hence, Claim~\ref{claim:type-surgery}, which extends unchanged to the
class-label encoding, yields that $\sigma$ extends $\tau$, satisfies
$\psi_t$, and
\(
  \TTTo(\sigma)
  \subseteq
  \TTTo(\tau)\cup\TTTo(\rho(\sigma')).
\)
Since $\tau$ is a $\cW''$-requirement,
$\TTTo(\tau)\subseteq\cW''$.
Since $\sigma'$ is a $(\cW',t)$-extension, $\TTTo(\sigma')\subseteq\cW'$, and therefore
\(
  \TTTo(\rho(\sigma'))\subseteq\rho(\cW')\subseteq\cW''.
\)
Thus $\TTTo(\sigma)\subseteq\cW''$, so $\sigma$ is a
$(\cW'',t)$-extension of $\tau$.

\paragraph{A finite witness.}
We now retain only finitely many class labels of each profile.
For every profile $P$, let
\( \cJ_P \coloneq \{j \in \N \mid P_j=P\}, \)
and choose $\cK_P \subseteq \cJ_P$ such that $|\cK_P|=\min\{|\cJ_P|,m+1\}$.

Write
\(
  \cK \coloneq \bigcup_P \cK_P.
\)
Let $\cW'''$ consist of the outer-types in $\cW''$ in which every
variable carries a class label $\LLL_j$ with $j\in\cK$.
Formally, for every $k\in[0,m+1]$ and every $k$-outer-type $\tau_\ro$,
\tightdisplay{\[
  \tau_\ro\in\cW'''
  \quad\Longleftrightarrow\quad
  \tau_\ro\in\cW''
  \text{ and }
  \tau_\ro\models
  \LLL_{j_1}(x_1)\wedge\dots\wedge\LLL_{j_k}(x_k)
  \text{ for some }j_1,\dots,j_k\in\cK.
\]}%
Note that class labels carried by constants are retained: if
$C_j\neq\emptyset$, then $\cJ_{P_j}=\{j\}$.

\begin{claim}\label{claim:EEE-third-witness}
  The set $\cW'''$ is a satisfiability witness for $\phi$.
\end{claim}
\begin{proof}
  It is straightforward to verify that $\cW'''$ is $\phi$-realisable.
  To verify the $\phi$-extension property, fix $t \in \cI$ and a $\cW'''$-requirement $\tau$.
  Since $\cW''' \subseteq \cW''$, the $m$-type $\tau$ is also a
  $\cW''$-requirement, and hence admits a $(\cW'',t)$-extension
  $\sigma$.
  Let $\LLL_j$ be the class label carried by the variable $x_{m+1}$ in $\sigma$, and let $P=P_j$.
  If $j \in \cK_P$, then $\sigma$ is already a $(\cW''',t)$-extension of $\tau$.

  Suppose otherwise.
  Then $|\cJ_P|>m+1$ and $|\cK_P|=m+1$.
  Since the variables occurring in $\tau$ carry at most
  $m$ distinct class labels, there exists
  \( k \in \cK_P \)
  that is not carried by any of them.
  As $P_j=P_k$, the transposition exchanging $j$ and $k$ is
  profile-preserving.
  Moreover, it fixes $\tau$, while replacing the label of $x_{m+1}$
  by the retained label $\LLL_k$.
  Since $\cW''$ is closed under profile-preserving bijections, the image
  of $\sigma$ under this transposition is a $(\cW'',t)$-extension of
  $\tau$.
  Since all labels occurring in $\tau$ are retained in $\cW'''$, it is in fact a $(\cW''',t)$-extension of $\tau$.
\end{proof}

\paragraph{Size bound.}
We now bound the size of $\cW'''$.
Each outer-type in $\cW'''$ is described by its $\Sigma_0$-reduct
and the class labels carried by its variables and constants.
For every $k\le m+1$, the number of all $k$-outer-types over
$\Sigma_0$ is
\(
  2^{2^{\cO(|\phi_{\rm in}|\log|\phi_{\rm in}|)}}.
\)
Since at most $m+1$ class labels are retained for each profile $P_j$,
the number of class labels occurring in $\cW'''$ is
at most $m+1$ times the bound in~\eqref{eqn:E-second-witness-bound},
and hence at most
\(
2^{2^{\cO(|\phi_{\rm in}|^2)}}.
\)
Therefore, $|\cW'''|$ is doubly exponential in $|\phi_{\rm in}|$.

We now take the reduct of $\cW'''$ obtained by forgetting the
class-label predicates, which are no longer needed.
The resulting collection can be viewed as a satisfiability witness over
the ordinary finite signature.
The random model construction of Subsection~\ref{subsec:criterion}
therefore applies without change.
Hence, Lemma~\ref{lemma:random-models} yields a doubly exponential
bound on model size.

\begin{theorem}\label{theorem:main-DMDK-EQ}
  The class $\DMDK + \EEE$ has the doubly exponential model property.
  The satisfiability problem for $\DMDK + \EEE$ is \TwoNExpTime-complete.
\end{theorem}

\section{Remarks on Equality}
\label{sec:remarks}

Extending the class $\DMK{}$ with equality makes satisfiability
undecidable, under the convention that equality atoms are treated
syntactically as ordinary atoms.
This follows directly from the undecidability of the prefix class
$\forall\forall\exists$ with equality~\cite{Gol84}.
In this section, we investigate syntactic restrictions on the use of
equality in $\DMK{}$ that preserve decidability of satisfiability.

\paragraph{Nominals.}
Let $\DMK+\cO$ denote the extension of $\DMK{}$ in which equality atoms
may contain at most one variable.
Thus, equality atoms may have the forms
$x=x$, $x=c$, and $c=c'$, where $x$ is a variable and $c,c'$ are
constants.
Such atoms are treated syntactically as ordinary atoms with at most one
variable.
For example, $\DMK+\cO$ contains sentences such as
\tightdisplay{\[
  \theta_{\forall,c} \;\coloneq\; \forall x\; \bigl(P(x)\rightarrow x=c\bigr),
\]}%
which requires $P$ to have at most one realisation.

Nevertheless, satisfiability for $\DMK+\cO$ is undecidable.
We show this by a reduction from the equality-free and constant-free prefix class
$\forall\exists\forall$, whose satisfiability is
undecidable~\cite{KMW62,BGG}.

Fix such a sentence
\(
  \phi
  =
  \forall x_1\; \exists x_2\; \forall x_3\; \psi,
\)
where $\psi$ is equality-free and constant-free.
For every relation symbol $R$ occurring in $\psi$, introduce a fresh
relation symbol $R^+$ whose arity is one greater than that of $R$.
Let $\psi^+$ be obtained from $\psi$ by replacing every atom
$R(\xs)$ by $R^+(\xs,x_4)$, where $x_4$ is fresh.
Letting $P$ be the same predicate as in the sentence $\theta_{\forall,c}$,
define
\tightdisplay{\[
  \chi
  \;\coloneq\;
  \forall x_1\; \exists x_2\; \forall x_3\; \exists x_4\;
  \bigl(P(x_4)\wedge\psi^+\bigr).
\]}%

The sentence $\chi$ belongs to $\DMK$ with no uniform variables.
Indeed, every atom of $\chi$ contains $x_4$ and hence has a $\chi$-prefix
ending with the existential quantifier $\exists x_4$.
Moreover, $\chi \wedge \theta_{\forall,c}$ is satisfiable if and only if $\phi$ is satisfiable.
Consequently, satisfiability for $\DMK+\cO$ is undecidable.

The preceding reduction crucially relies on introducing relation symbols
of arity at least $3$.
Indeed,
for a binary relation symbol $R$, the formula
\(
  \exists y\; \bigl(y=c\wedge R(x,y)\bigr)
\)
is equivalent to $R(x,c)$.
The latter contains only one variable, so its $\phi$-prefix has length
at most~$1$ and therefore satisfies the restrictions of $\DMK{}$.
This is not a coincidence, as shown by the following proposition.

\begin{proposition}\label{prop:DMDK-nominals-binary}
  The satisfiability problem for $\DMK+\cO$ is decidable over signatures containing
  only unary and binary relation symbols.
\end{proposition}

\begin{proof}
  Fix a sentence $\phi\in\DMK+\cO$ over a signature
  $\Sigma=\langle\Rels,\Cons\rangle$ such that every relation symbol in
  $\Rels$ is unary or binary.
  Introduce a fresh unary relation symbol $\CC$, intended to distinguish
  the elements named by constants from the unnamed elements.

  We first conjoin to $\phi$ the sentence
  \(
    \bigwedge_{c\in\Cons}\CC(c),
  \)
  asserting that every constant satisfies $\CC$, and then systematically
  replace each subformula of the form
  \(
    \exists x\;\psi
  \)
  by
  \tightdisplay{\[
    \exists x\;(\neg\CC(x)\wedge\psi)
    \;\vee\;
    \bigvee\nolimits_{c\in\Cons}\psi[c/x],
  \]}%
  and each subformula of the form
  \(
    \forall x\;\psi
  \)
  by
  \tightdisplay{\[
    \forall x\;(\neg\CC(x)\rightarrow\psi)
    \;\wedge\;
    \bigwedge\nolimits_{c\in\Cons}\psi[c/x].
  \]}%
  Thus, every quantified variable is explicitly split according to
  whether it is interpreted by an unnamed element or by a constant.

  After this expansion, all equality atoms can be eliminated.
  Atoms of the forms $x=c$ and $c=c'$ with $c\neq c'$ are replaced by
  $\bot$, while atoms of the forms $x=x$ and $c=c$ are replaced by
  $\top$.
  Indeed, an atom $x=c$ occurs only in a part of the formula in which $x$
  ranges over unnamed elements, while distinct constants denote distinct
  elements by the standard name assumption.

  Let $\phi'$ be the resulting equality-free sentence.
  Then $\phi'$ belongs to $\DMK{}$.
  Indeed, substituting a constant for a variable in an atom of arity at
  most $2$ leaves at most one variable in that atom.
  Hence, every new $\phi'$-prefix arising from such a substitution has
  length at most~$1$.

  By construction, $\phi$ is satisfiable if and only if
  \(
    \phi'
  \)
  admits a model in which no unnamed element satisfies $\CC$.
  This additional requirement can be dropped.
  Indeed, all quantifiers in $\phi'$ are relativised to $\neg\CC$, so
  unnamed elements satisfying $\CC$ can simply be discarded without
  affecting satisfaction of $\phi'$.
\end{proof}

Proposition~\ref{prop:DMDK-nominals-binary} is particularly relevant
from the perspective of description logics;
see~\cite{Baader_Horrocks_Lutz_Sattler_2017} for a comprehensive
introduction.
Their signatures typically consist of unary relation symbols,
representing concepts, and binary relation symbols, representing roles.
Over such signatures, $\DMK{}$ remains quite expressive.
For example, it contains sentences such as
\tightdisplay{\[
  \forall x_1\; \forall x_2\; \forall x_3\; \exists x_4\;
  \bigl(
    E(x_1,x_4)
    \wedge E(x_2,x_4)
    \wedge E(x_3,x_4)
  \bigr),
\]}%
which requires every three elements to have a common $E$-successor.
The extension $\DMK+\cO$ additionally allows \emph{nominals}, that is,
singleton concepts, to be expressed.

\paragraph{Equality in prenex conjuncts.}
We now consider a different approach to incorporating equality into
$\DMK{}$.
Rather than allowing equality atoms directly in $\DMK{}$ sentences, we
ask whether decidability can be preserved by conjoining a sentence of
$\DMK{}$ with a sentence from some fragment $\cF$ with equality.
Thus, we consider sentences of the form
\(\mu\wedge\nu\) with $\mu \in \DMK{}$ and $\nu\in\cF$. 

To obtain a meaningful classification, we restrict attention to prenex
classes.
The reduction establishing undecidability of $\DMK+\cO$ already shows
that undecidability arises when $\DMK{}$ is extended by a single fixed
sentence $\theta_{\forall,c}$ with prefix $\forall$ and one constant.
The constant can be avoided at the cost of a slightly longer prefix:
undecidability persists for a single fixed sentence with prefix
$\forall\forall$ or $\exists\forall$.
Indeed, in the reduction for $\DMK+\cO$, it suffices to replace
$\theta_{\forall,c}$ by either of the sentences enforcing that $P$ has at most one realisation:
\tightdisplay{\[
  \theta_{\forall\forall}
  \coloneq
  \forall x_1\; \forall x_2\;
  \bigl(
    P(x_1)\wedge P(x_2)
    \rightarrow x_1=x_2
  \bigr),
  \qquad
  \theta_{\exists\forall}
  \coloneq
  \exists x_1\; \forall x_2\;
  \bigl(
    P(x_2)\rightarrow x_1=x_2
  \bigr).
\]}%

For a sentence $\theta$, let $\DMK \wedge \theta$ denote the class of sentences
of the form
$\mu\wedge\theta$ with $\mu$ ranging $\DMK{}$.
By the argument above, satisfiability for each of
$\DMK \wedge \theta_{\forall,c}$,
$\DMK \wedge \theta_{\forall\forall}$, and
$\DMK \wedge \theta_{\exists\forall}$ is undecidable, with the latter two
results holding even over constant-free signatures.

By the considerations above, the only remaining choices for $\cF$ among
prenex classes are $\exists^*$ over signatures with constants and
$\forall\exists^*$ over constant-free signatures.
The former case is immediate:
we can replace every variable by a fresh constant,
and subsequently eliminate equality using the standard name assumption.
We therefore focus on the latter.
As we show next, this case yields a decidable formalism as well.

Let $[\forall\exists^*]$ denote the class of sentences of the form
\(
  \forall x_1\; \exists x_2\; \cdots\; \exists x_m\; \psi,
\)
where $m\ge1$ and $\psi$ is quantifier-free, equality-free, and
constant-free.
Let $[\forall\exists^*]_=$ denote the analogous class in which equality
is allowed in $\psi$.
We first observe that equality can be eliminated from
$[\forall\exists^*]_=$.

\begin{lemma}\label{lem:forall-exists-equality-elimination}
  Let $\Sigma=\langle\Rels,\emptyset\rangle$ be a signature, and let
  $E$ be a fresh binary relation symbol.
  For every sentence $\phi\in[\forall\exists^*]_=$ over $\Sigma$, one
  can construct a sentence $\phi_E\in[\forall\exists^*]$ over
  $\Sigma\cup\{E\}$ such that $\phi_E \models \phi$ and 
  every model of $\phi$ admits a
  $(\Sigma\cup\{E\})$-expansion that satisfies $\phi_E$.
\end{lemma}

The lemma follows by setting $p=0$ in
\cite[Proposition~6.3.31]{BGG}, where $p$ denotes the number of leading
existential quantifiers in the prefix class
$\exists^p\forall\exists^*$.
Although the proposition itself states a weaker condition
than $\phi_E\models\phi$,
in the special case $p=0$, its proof in fact yields this
entailment.

We now define the class of interest
$\DMDK+[\forall\exists^*]_=$ to consist of sentences of the form
\tightdisplay{\[
  \phi_1\wedge\dots\wedge\phi_k,
  \qquad
  \text{where}
  \qquad
  k\ge1
  \text{ and }
  \phi_1,\dots,\phi_k\in\DMK{}\cup[\forall\exists^*]_=.
\]}%

By Lemma~\ref{lem:forall-exists-equality-elimination}, every
$[\forall\exists^*]_=$ conjunct can be replaced by an equality-free
$[\forall\exists^*]$ conjunct, and the latter class is a syntactic
subfragment of $\DMK{}$.
More precisely, we apply the lemma separately to each
$[\forall\exists^*]_=$ conjunct, introducing a fresh binary relation
symbol $E$ each time.
This yields an equisatisfiable conjunction of $\DMK{}$ sentences, and
thus a reduction to satisfiability for $\DMDK{}$.

\begin{proposition}\label{prop:DMDK-plus-forall-exists}
  The satisfiability problem for $\DMDK+[\forall\exists^*]_=$ is
  decidable over signatures without constant symbols.
\end{proposition}

\paragraph{Acknowledgment.}
The author thanks Emanuel Kiero\'nski for discussions,
and the reviewers of LPAR 2026 for their comments.
The author is supported by Polish National Science Center, grant No. 2021/41/B/ST6/00996.

\bibliographystyle{plain}
\bibliography{master}

\appendix

\section{Missing Proofs from Subsection~\ref{subsec:EEE-lower}}
\label{appendix:equivalences}

\subsection{Proof of Proposition~\ref{prop:EEE-lower}}

In this subsection, we establish that satisfiability for
$\MON{}+\EEE$ is $\TwoNExpTime{}$-hard.
The corresponding $\TwoNExpTime{}$ lower bound for $\DMK{}+\EEE$ can
also be inferred from the work of Kiero\'nski and
Kuusisto~\cite[Section~5.1]{KK15} on the uniform one-dimensional
fragment with one equivalence relation.
Our construction is similar in spirit, but yields the stronger result
that all non-distinguished relation symbols can be restricted to unary
predicates.

We reduce from the \emph{doubly exponential tiling problem}, which is
$\TwoNExpTime{}$-complete.
An instance of this problem is a tuple
\(
  \cT
  =
  \langle
    \cC,c_I,c_F,\cH,\cV,n
  \rangle,
\)
where $\cC$ is a non-empty finite set of \emph{colours},
$c_I,c_F \in \cC$ are the \emph{initial} and \emph{final} colours,
$\cH,\cV \subseteq \cC \times \cC$ are the sets of admissible
\emph{horizontal} and \emph{vertical} pairs, respectively, and
$n \in \N$ is given in unary.
The instance $\cT$ is \emph{solvable} if there exists a tiling
\(
  f \colon
  [0,2^{2^n}-1] \times [0,2^{2^n}-1]
  \longrightarrow
  \cC
\)
such that:
\begin{itemize}[nosep]
  \item
    $f(0,0)=c_I$ and
    $f(2^{2^n}-1,2^{2^n}-1)=c_F$;
  \item
    $\langle f(i,j),f(i+1,j)\rangle \in \cH$
    for every $i \in [0,2^{2^n}-2]$ and
    $j \in [0,2^{2^n}-1]$;
  \item
    $\langle f(i,j),f(i,j+1)\rangle \in \cV$
    for every $i \in [0,2^{2^n}-1]$ and
    $j \in [0,2^{2^n}-2]$.
\end{itemize}
The doubly exponential tiling problem asks whether a given instance
$\cT$ is solvable.

Fix an instance
$\cT=\langle\cC,c_I,c_F,\cH,\cV,n\rangle$.
We construct, in time polynomial in the size of $\cT$, a sentence
\(
  \phi_\cT \in \MON{} + \EEE
\)
such that $\cT$ is solvable if and only if $\phi_\cT$ is satisfiable.

The signature of $\phi_\cT$ is constant-free and, besides the
distinguished equivalence symbol $\EEE$, contains the following unary
relation symbols: $P_i$ ($i \in [0,n-1]$), $F_c$ ($c \in \cC$), $X$, $Y$.

We naturally view the predicates $P_0,\dots,P_{n-1}$ as bits, thereby
associating with each domain element an integer in the range
$[0,2^n-1]$.
More precisely, for a domain element $a$, let
\tightdisplay{\[
  p_a
  \coloneq
  \sum_{i=0}^{n-1} p_{a,i}\cdot 2^i,
  \qquad\text{where}\qquad
  p_{a,i}
  \coloneq
  \begin{cases}
    1 & \text{if $P_i(a)$ holds,}\\
    0 & \text{otherwise.}
  \end{cases}
\]}%

Let $\psi^P_i$, $\psi^P_{+1}$, $\psi^P_=$, and $\psi^P_<$ be the
standard formulas, of size polynomial in $n$, expressing
\tightdisplay{\[
  \psi^P_i(a)
    \;\Longleftrightarrow\; p_a{=}i,
  \quad
  \psi^P_{+1}(a,b)
    \;\Longleftrightarrow\; p_a{+}1{=}p_b,
  \quad
  \psi^P_=(a,b)
    \;\Longleftrightarrow\; p_a{=}p_b,
  \quad
  \psi^P_<(a,b)
    \;\Longleftrightarrow\; p_a{<}p_b.
\]}%

We first ensure that every equivalence class contains a representative of each possible value of $p_a$, ranging from $0$ to $2^n-1$:
\tightdisplay{\begin{equation}\label{eq:tiling-generate-positions}
  \forall x\,
  \Bigl(
    \exists y\,
      \bigl(x \EEE y \wedge \psi^P_0(y)\bigr)
    \;\wedge\;
    \bigl(
      \neg\psi^P_{2^n-1}(x)
      \rightarrow
      \exists y\,
        \bigl(x \EEE y \wedge \psi^P_{+1}(x,y)\bigr)
    \bigr)
  \Bigr).
\end{equation}}%

We further require that elements of the same equivalence class
representing the same integer agree on the predicates $X$ and $Y$:
\tightdisplay{\begin{equation}\label{eq:tiling-consistent-X-and-Y}
  \forall x,y\,
  \Bigl(
    \bigl(x \EEE y \wedge \psi^P_=(x,y)\bigr)
    \rightarrow
    \bigl(
      (X(x)\leftrightarrow X(y))
      \wedge
      (Y(x)\leftrightarrow Y(y))
    \bigr)
  \Bigr).
\end{equation}}%

We can now associate with every equivalence class of $\EEE$ a pair of
integers in the range $[0,2^{2^n}-1]$.
Let $\cE$ be an equivalence class of the interpretation of $\EEE$.
We define
\tightdisplay{\[
  X_\cE
  \coloneq
  \sum_{i=0}^{2^n-1} X_{\cE,i}\cdot 2^i,
  \qquad\text{where}\qquad
  X_{\cE,i}
  \coloneq
  \begin{cases}
    1
      & \text{if $X(a)$ holds for some $a\in\cE$ with $p_a=i$,}\\
    0
      & \text{otherwise.}
  \end{cases}
\]}%
The integer $Y_\cE$ is defined analogously, using $Y$ in place of $X$.

Every value $i\in[0,2^n-1]$ is represented in $\cE$, and any two
elements of $\cE$ representing the same value agree on both $X$ and
$Y$.
Hence, $X_\cE$ and $Y_\cE$ are well-defined.
We call them the \emph{$X$-coordinate} and \emph{$Y$-coordinate} of
$\cE$, respectively, and view the pair
$\langle X_\cE,Y_\cE\rangle$ as a position in the doubly exponential
grid.
For convenience, we extend the notions of $X$-coordinate and $Y$-coordinate to domain elements:
for each $a \in \cE$,
we set
$X_a \coloneq X_\cE$
and
$Y_a \coloneq Y_\cE$.

Under this convention, we define formulas
$\psi^X_0$, $\psi^X_{2^{2^n}-1}$, $\psi^X_{+1}$, and $\psi^X_=$
expressing
\tightdisplay{\[
\begin{array}{rcl@{\qquad}rcl}
  \psi^X_0(a)
    &\Longleftrightarrow& X_a=0,
  &
  \psi^X_{2^{2^n}-1}(a)
    &\Longleftrightarrow& X_a=2^{2^n}-1,
  \\[2mm]
  \psi^X_{+1}(a,b)
    &\Longleftrightarrow& X_a+1=X_b,
  &
  \psi^X_=(a,b)
    &\Longleftrightarrow& X_a=X_b.
\end{array}
\]}%

The formulas for the two extreme values and for equality of coordinates are straightforward:
\tightdisplay{\begin{align*}
  &\psi^X_0(x)
  \;\coloneq\;
  \forall y\,
  \bigl(x \EEE y \rightarrow \neg X(y)\bigr),
  \qquad
  \psi^X_{2^{2^n}-1}(x)
  \;\coloneq\;
  \forall y\,
  \bigl(x \EEE y \rightarrow X(y)\bigr), \\
  &\psi^X_=(x,y)
    \;\coloneq\;
    \forall x',y'\,
    \Bigl(
      \bigl(x \EEE x'
      \wedge y \EEE y'
      \wedge \psi^P_=(x',y')\bigr)
      \rightarrow
      \bigl(X(x')\leftrightarrow X(y')\bigr)
    \Bigr).
\end{align*}}%

The increment formula is slightly more involved.
We first define a bit-level auxiliary formula
\tightdisplay{\[
  \psi^X_{+1,\mathrm{bit}}(x',y')
  \;\coloneq\;
  \Bigl(
    \exists x''\,
    \bigl(
      x' \EEE x''
      \wedge
      \psi^P_<(x'',x')
      \wedge
      \neg X(x'')
    \bigr)
  \Bigr)
  \leftrightarrow
  \bigl(X(x')\leftrightarrow X(y')\bigr).
\]}%
Using this auxiliary formula, we set
\tightdisplay{\[
  \psi^X_{+1}(x,y)
  \;\coloneq\;
  \neg \psi^X_{2^{2^n}-1}(x)
  \;\wedge\;
  \forall x',y'\,
  \Bigl(
    \bigl(x \EEE x'
    \wedge
    y \EEE y'
    \wedge
    \psi^P_=(x',y')\bigr)
    \rightarrow
    \psi^X_{+1,\mathrm{bit}}(x',y')
  \Bigr).
\]}%

To see why this expresses incrementation, recall the standard binary
increment operation.
When adding one to an integer, the least significant $0$-bit is
changed to $1$, all less significant bits are changed from $1$ to $0$,
and all more significant bits remain unchanged.
Equivalently, a bit at position $i$ remains unchanged precisely when
there is a $0$-bit at some lower-order position $j<i$.
The formula $\psi^X_{+1,\mathrm{bit}}$ expresses exactly this
condition for each bit position, while
$\neg\psi^X_{2^{2^n}-1}(x)$ excludes the maximal value and hence prevents overflow.

The formulas
$\psi^Y_0$, $\psi^Y_{2^{2^n}-1}$, $\psi^Y_{+1}$, and $\psi^Y_=$
for the $Y$-coordinate are defined analogously by replacing $X$ by $Y$.

Having defined the formulas above, we can now encode a
solution to the tiling instance $\cT$.
We first ensure that every element is assigned exactly one colour:
\tightdisplay{\begin{equation}\label{eq:tiling-unique-colour}
  \forall x\,
  \bigvee\nolimits_{c\in\cC}
  \Bigl(
    F_c(x)
    \wedge
    \bigwedge\nolimits_{c'\in\cC\setminus\{c\}}
      \neg F_{c'}(x)
  \Bigr).
\end{equation}}%
Since different $\EEE$-classes may represent the same position in the
grid, we further require that any two representatives of the same
position have the same colour:
\tightdisplay{\begin{equation}\label{eq:tiling-colour-consistency}
  \forall x,y\,
  \Bigl(
    \psi^X_=(x,y)
    \wedge
    \psi^Y_=(x,y)
    \rightarrow
    \bigwedge\nolimits_{c\in\cC}
      \bigl(F_c(x)\leftrightarrow F_c(y)\bigr)
  \Bigr).
\end{equation}}%
Thus, each represented pair of coordinates has a well-defined colour.

We next require the lower-left and upper-right corners of the grid to
have the prescribed initial and final colours:
\tightdisplay{\begin{equation}\label{eq:tiling-initial-and-final}
  \exists x\,
  \bigl(
    \psi^X_0(x)
    \wedge
    \psi^Y_0(x)
    \wedge
    F_{c_I}(x)
  \bigr)
  \quad\wedge\quad
  \exists x\,
  \bigl(
    \psi^X_{2^{2^n}-1}(x)
    \wedge
    \psi^Y_{2^{2^n}-1}(x)
    \wedge
    F_{c_F}(x)
  \bigr).
\end{equation}}%

Finally, we ensure the existence of appropriately coloured horizontal
and vertical successors.
Every point whose $X$-coordinate is not maximal has a point immediately
to its right, with the same $Y$-coordinate and a horizontally
compatible colour:
\tightdisplay{\begin{equation}\label{eq:tiling-horizontal}
  \forall x\,
  \Bigl(
    \neg\psi^X_{2^{2^n}-1}(x)
    \rightarrow
    \exists y\,
    \Bigl(
      \psi^X_{+1}(x,y)
      \wedge
      \psi^Y_=(x,y)
      \wedge
      \bigvee\nolimits_{\langle c,c'\rangle\in\cH}
        \bigl(F_c(x)\wedge F_{c'}(y)\bigr)
    \Bigr)
  \Bigr).
\end{equation}}%
Similarly, every point whose $Y$-coordinate is not maximal has a point
immediately above it, with the same $X$-coordinate and a vertically
compatible colour:
\tightdisplay{\begin{equation}\label{eq:tiling-vertical}
  \forall x\,
  \Bigl(
    \neg\psi^Y_{2^{2^n}-1}(x)
    \rightarrow
    \exists y\,
    \Bigl(
      \psi^X_=(x,y)
      \wedge
      \psi^Y_{+1}(x,y)
      \wedge
      \bigvee\nolimits_{\langle c,c'\rangle\in\cV}
        \bigl(F_c(x)\wedge F_{c'}(y)\bigr)
    \Bigr)
  \Bigr).
\end{equation}}%

Let $\phi_\cT$ be the conjunction of sentences~\eqref{eq:tiling-generate-positions}--\eqref{eq:tiling-vertical}.
It is satisfiable if and only if the tiling instance $\cT$ is solvable.
Moreover, $\phi_\cT$ belongs to $\MON{} + \EEE$ and can be constructed from $\cT$ in polynomial time.
This gives the desired reduction from the doubly exponential tiling problem to satisfiability for $\MON{}+\EEE$.
Therefore, satisfiability for $\MON{}+\EEE$ is $\TwoNExpTime{}$-hard.

\subsection{Proof of Proposition~\ref{prop:EEE-undec}}

In this subsection,
we establish that satisfiability for $\MON{}+2\EEE$ and $\MON{}+\TRrefl$ is co-$\RE$-complete.
Membership in co-$\RE$ is immediate: the
distinguished relations are finitely axiomatisable in first-order
logic, so unsatisfiability reduces to first-order unsatisfiability and
is therefore recursively enumerable.
It thus remains to establish co-$\RE$-hardness.

For this, we build on the reductions of Kiero\'nski and
Kuusisto~\cite[Section~5.5]{KK15}.
We first explain how their construction for $\TR$ can be modified for
the case with $\TRrefl$, and then give details of the corresponding
construction for $2\EEE$, which are not spelled out
in~\cite{KK15}.

We reduce from the \emph{infinite tiling problem}, which is co-$\RE$-complete.
An instance of this problem is a tuple
\(
  \cT
  =
  \langle
    \cC,c_I,\cH,\cV
  \rangle,
\)
where $\cC$ is a non-empty finite set of \emph{colours},
$c_I \in \cC$ is the \emph{initial} colour,
and $\cH,\cV \subseteq \cC \times \cC$ are the sets of admissible
\emph{horizontal} and \emph{vertical} pairs, respectively.
The instance $\cT$ is \emph{solvable} if there exists a tiling
\(
  f \colon
  \N \times \N
  \longrightarrow
  \cC
\)
such that:
\begin{itemize}[nosep]
  \item
    $f(0,0)=c_I$;
  \item
    $\langle f(i,j),f(i+1,j)\rangle \in \cH$
    for every $i,j \in \N$;
  \item
    $\langle f(i,j),f(i,j+1)\rangle \in \cV$
    for every $i,j \in \N$.
\end{itemize}
The infinite tiling problem asks whether a given instance
$\cT$ is solvable.

\paragraph{Reflexive-transitive relation.}
Fix an instance
$\cT=\langle\cC,c_I,\cH,\cV\rangle$.
We construct, effectively from $\cT$, a sentence
\(
  \phi_\cT\in\MON{}+\TRrefl
\)
such that $\cT$ is solvable if and only if $\phi_\cT$ is satisfiable.

Besides the distinguished symbol $\TRrefl$, the signature of
$\phi_\cT$ contains the following unary relation symbols:
$X_i$ ($i\in[0,3]$), $Y_j$ ($j\in[0,3]$), and
$F_c$ ($c\in\cC$).

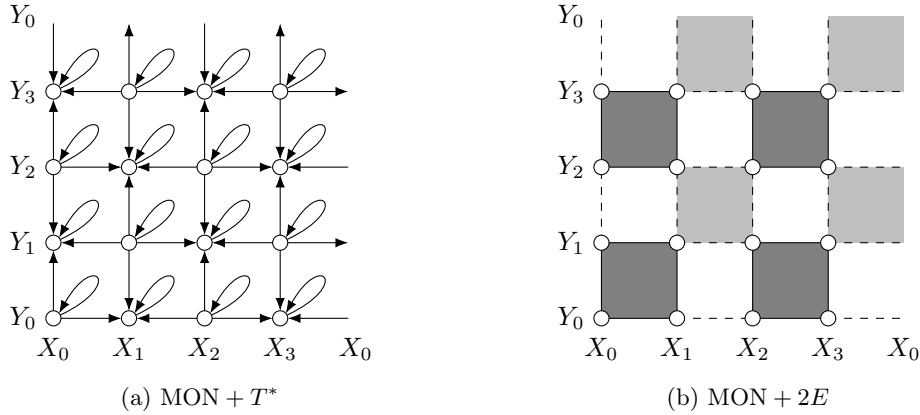
\begin{figure}[h]
  \centering
  \begin{subfigure}{.5\textwidth}
    \centering
    \begin{tikzpicture}
  \foreach \x in {0,1,2,3} {
    \foreach \y in {0,1,2,3} {
      \begin{scope}[shift={(\x,\y)}]
        \draw (0,0) circle (0.1cm);
        \draw[rotate=-48, arrows=-Latex] (0.03,0.1) .. controls (0.3,1) and (-0.3,1) .. (-0.0,0.1);
      \end{scope}
    }
  }
  \foreach \x in {0,2} {
    \foreach \y in {0,2} {
      \begin{scope}[shift={(\x,\y)}]
        \draw[arrows=-Latex] (0.1,0) -- (0.9,0);
        \draw[arrows=-Latex] (0,0.1) -- (0,0.9);
      \end{scope}
    }
  }
  \foreach \x in {1,3} {
    \foreach \y in {1,3} {
      \begin{scope}[shift={(\x,\y)}]
        \draw[arrows=-Latex] (0.1,0) -- (0.9,0);
        \draw[arrows=-Latex] (0,0.1) -- (0,0.9);
      \end{scope}
    }
  }
  \foreach \x in {1,3} {
    \foreach \y in {0,2} {
      \begin{scope}[shift={(\x,\y)}]
        \draw[arrows=Latex-] (0.1,0) -- (0.9,0);
        \draw[arrows=Latex-] (0,0.1) -- (0,0.9);
      \end{scope}
    }
  }
  \foreach \x in {0,2} {
    \foreach \y in {1,3} {
      \begin{scope}[shift={(\x,\y)}]
        \draw[arrows=Latex-] (0.1,0) -- (0.9,0);
        \draw[arrows=Latex-] (0,0.1) -- (0,0.9);
      \end{scope}
    }
  }

  \node at (0,-0.4) {$X_0$};
  \node at (1,-0.4) {$X_1$};
  \node at (2,-0.4) {$X_2$};
  \node at (3,-0.4) {$X_3$};
  \node at (4,-0.4) {$X_0$};

  \node at (-0.4,0) {$Y_0$};
  \node at (-0.4,1) {$Y_1$};
  \node at (-0.4,2) {$Y_2$};
  \node at (-0.4,3) {$Y_3$};
  \node at (-0.4,4) {$Y_0$};
\end{tikzpicture}
    \label{fig:grid-trans}
    \caption{$\MON+\TRrefl$}
  \end{subfigure}%
  \begin{subfigure}{.5\textwidth}
    \centering
    \begin{tikzpicture}
  \foreach \x in {0,2} {
    \foreach \y in {0,2} {
      \begin{scope}[shift={(\x,\y)}]
        \draw[fill=gray] (0,0) rectangle (1,1);
      \end{scope}
    }
  }
  \foreach \x in {1,3} {
    \foreach \y in {1,3} {
      \begin{scope}[shift={(\x,\y)}]
        \draw[fill=lightgray, draw=none] (0,0) rectangle (1,1);
      \end{scope}
    }
  }
  \foreach \x in {1,3} {
    \foreach \y in {0,1,2,3} {
      \begin{scope}[shift={(\x,\y)}]
        \draw[dashed] (0,0) -- (1,0);
      \end{scope}
    }
  }
  \foreach \x in {0,1,2,3} {
    \foreach \y in {1,3} {
      \begin{scope}[shift={(\x,\y)}]
        \draw[dashed] (0,0) -- (0,1);
      \end{scope}
    }
  }
  \foreach \x in {0,1,2,3} {
    \foreach \y in {0,1,2,3} {
      \begin{scope}[shift={(\x,\y)}]
        \draw[fill=white] (0,0) circle (0.1cm);
      \end{scope}
    }
  }

  \node at (0,-0.4) {$X_0$};
  \node at (1,-0.4) {$X_1$};
  \node at (2,-0.4) {$X_2$};
  \node at (3,-0.4) {$X_3$};
  \node at (4,-0.4) {$X_0$};

  \node at (-0.4,0) {$Y_0$};
  \node at (-0.4,1) {$Y_1$};
  \node at (-0.4,2) {$Y_2$};
  \node at (-0.4,3) {$Y_3$};
  \node at (-0.4,4) {$Y_0$};
\end{tikzpicture}
    \label{fig:grid-twoeq}
    \caption{$\MON+2\EEE$}
  \end{subfigure}
  \caption{Grid structures.}
  \label{fig:grids}
\end{figure}

The intended structure representing the infinite
$\N\times\N$ grid is depicted on the left in Figure~\ref{fig:grids}.
Each element represents a single tile at some position $(i,j)$.
It satisfies precisely one predicate $X_{i'}$, where
$i'\equiv i\pmod 4$, and precisely one predicate $Y_{j'}$, where
$j'\equiv j\pmod 4$. 
The pairs related by $\TRrefl$ are indicated by arrows. In particular,
every element has a self-loop, as required by reflexivity.
The relation is arranged so that
there are no directed paths of length two consisting entirely of
non-reflexive edges, and therefore transitivity forces no additional connections.

We now describe this structure by a sentence in $\MON{}+\TRrefl$.
First, we require every element to satisfy precisely one $X_i$ and
precisely one $Y_j$:
\tightdisplay{\begin{equation}
  \forall x\,
  \bigvee\nolimits_{i,j\in[0,3]}
  \Bigl(
    X_i(x)
    \wedge
    Y_j(x)
    \wedge
    \bigwedge\nolimits_{i'\neq i} \neg X_{i'}(x)
    \wedge
    \bigwedge\nolimits_{j'\neq j} \neg Y_{j'}(x)
  \Bigr).
\end{equation}}%

We now define a formula expressing that $y$ is a horizontal successor
of $x$.
For $k\in\N$, write $\TR_k^*(x,y)$ for $\TRrefl(x,y)$ when $k$ is even and for $\TRrefl(y,x)$
when $k$ is odd.
Then set
\tightdisplay{\[
  \psi_H(x,y)
  \;\coloneq\;
  \bigvee\nolimits_{i,j\in[0,3]}
  \Bigl(
    X_i(x)
    \wedge X_{i+1}(y)
    \wedge Y_j(x)
    \wedge Y_j(y)
    \wedge \TR_{i+j}^*(x,y)
    \wedge \neg\TR_{i+j}^*(y,x)
  \Bigr),
\]}%
where the subscript of $X_{i+1}$ is taken modulo~$4$.

The formula $\psi_V(x,y)$ expressing that $y$ is a vertical successor
of $x$ is defined analogously: the $Y$-index is increased by one
modulo~$4$, while the $X$-index remains unchanged.

We also introduce an auxiliary formula expressing that two elements
represent the same tile:
\tightdisplay{\[
  \psi_E(x,y)
  \;\coloneq\;
  \TRrefl(x,y)\wedge\TRrefl(y,x).
\]}%
Since $\TRrefl$ is reflexive and transitive, the relation defined by
$\psi_E$ is an equivalence relation.

We now describe the grid connections.
First, elements representing the same tile are required to agree on all
unary predicates, that is, on the predicates $X_i$, $Y_j$, and $F_c$:
\tightdisplay{\begin{equation}\label{eqn:infinite-tiling-agree-on-all}
  \forall x,y\;
  \Bigl(
    \psi_E(x,y)
    \rightarrow
    \bigwedge\nolimits_{\text{$P$ is unary}}
    \bigl(P(x)\leftrightarrow P(y)\bigr)
  \Bigr).
\end{equation}}%

Next, horizontal successors are required to be unique up to
$\psi_E$-equivalence:
\tightdisplay{\begin{equation}\label{eqn:infinite-tiling-unique-H}
  \forall x,y,z\;
  \Bigl(
    \bigl(\psi_H(x,y)\wedge\psi_H(x,z)\bigr)
    \rightarrow
    \psi_E(y,z)
  \Bigr).
\end{equation}}%
Likewise, vertical successors are required to be unique up to
$\psi_E$-equivalence:
\tightdisplay{\begin{equation}\label{eqn:infinite-tiling-unique-V}
  \forall x,y,z\;
  \Bigl(
    \bigl(\psi_V(x,y)\wedge\psi_V(x,z)\bigr)
    \rightarrow
    \psi_E(y,z)
  \Bigr).
\end{equation}}%

Finally, we require every element to be the lower-left corner of a grid square:
\tightdisplay{\begin{equation}\label{eqn:infinite-tiling-square}
  \forall x\;
  \exists y,z,t\;
  \bigl(
    \psi_H(x,y)
    \wedge
    \psi_V(y,z)
    \wedge
    \psi_H(t,z)
    \wedge
    \psi_V(x,t)
  \bigr).
\end{equation}}%

Let us briefly explain why these conditions suffice to enforce the
intended grid structure.
The key point is that the squares provided
by~\eqref{eqn:infinite-tiling-square} can be consistently arranged
into a grid.

Fix an element $a_x$, and let $a_y,a_z,a_t$ be witnesses
for~\eqref{eqn:infinite-tiling-square}.
Set $b_x\coloneq a_y$, and let $b_y,b_z,b_t$ be witnesses for the same
formula with $b_x$ in place of $a_x$.
Then both $a_z$ and $b_t$ are vertical successors of
$a_y=b_x$.
Hence, by~\eqref{eqn:infinite-tiling-unique-V},
$\psi_E(a_z,b_t)$ holds.
Thus, up to $\psi_E$-equivalence, the square
$(b_x,b_y,b_z,b_t)$ shares its left vertical side with the right
vertical side of $(a_x,a_y,a_z,a_t)$.

Analogously, using~\eqref{eqn:infinite-tiling-unique-H}, squares can be
joined along their horizontal sides.
Iterating this construction therefore yields a representation of the
$\N\times\N$ grid, up to replacing individual grid points by
$\psi_E$-equivalent elements.
By~\eqref{eqn:infinite-tiling-agree-on-all}, such replacements preserve
all $X_i$, $Y_j$, and $F_c$ predicates, and hence do not affect either
the grid coordinates modulo~$4$ or the assigned colours.

It remains to impose the tiling constraints, which can be expressed in
the standard way.
First, every element is assigned precisely one colour and the initial colour is realised:
\tightdisplay{\begin{equation}
  \forall x\,
  \bigvee\nolimits_{c\in\cC}
  \Bigl(
    F_c(x)
    \wedge
    \bigwedge\nolimits_{c'\in\cC\setminus\{c\}}
      \neg F_{c'}(x)
  \Bigr)
  \;\wedge\;
  \exists x\; F_{c_I}(x).
\end{equation}}%
The horizontal and vertical compatibility conditions are expressed,
respectively, by
\tightdisplay{\begin{equation}
  \forall x\; \exists y\;
  \Bigl(
    \psi_H(x,y)
    \wedge
    \bigvee\nolimits_{\langle c,c'\rangle\in\cH}
      \bigl(F_c(x)\wedge F_{c'}(y)\bigr)
  \Bigr)
\end{equation}}%
and
\tightdisplay{\begin{equation}
  \forall x\; \exists y\;
  \Bigl(
    \psi_V(x,y)
    \wedge
    \bigvee\nolimits_{\langle c,c'\rangle\in\cV}
      \bigl(F_c(x)\wedge F_{c'}(y)\bigr)
  \Bigr).
\end{equation}}%

\paragraph{Two equivalence relations.}
Fix an instance
$\cT=\langle\cC,c_I,\cH,\cV\rangle$.
We construct, effectively from $\cT$, a sentence
\(
  \phi_\cT\in\MON{}+2\EEE
\)
such that $\cT$ is solvable if and only if $\phi_\cT$ is satisfiable.

We adapt the construction given above for $\MON{}+\TRrefl$.
Let $E_1,E_2$ denote the two distinguished equivalence symbols of
$\MON{}+2\EEE$, and use the same unary predicates $X_i$, $Y_j$, and
$F_c$ as before.

The intended structure representing the infinite $\N\times\N$ grid is
depicted on the right in Figure~\ref{fig:grids}.
The dark grey regions represent $E_2$-classes, while the light grey
regions represent $E_1$-classes.
The classes alternate so that every pair of horizontally or vertically
adjacent grid points belongs to exactly one common equivalence class,
with the choice of the equivalence relation determined by the parity
of the corresponding coordinate.
Moreover, the intersection of an $E_1$-class and an $E_2$-class
contains precisely one grid point in the intended structure.

As before, we first require every element to satisfy precisely one
predicate $X_i$ and precisely one predicate $Y_j$.
For $k\in\N$, write $E_k$ for $E_1$ when $k$ is odd and for $E_2$
when $k$ is even.
We then define the horizontal-successor formula by
\tightdisplay{\[
  \psi_H(x,y)
  \;\coloneq\;
  \bigvee\nolimits_{i,j\in[0,3]}
  \Bigl(
    X_i(x)
    \wedge X_{i+1}(y)
    \wedge Y_j(x)
    \wedge Y_j(y)
    \wedge E_i(x,y)
    \wedge \neg E_{i+1}(x,y)
  \Bigr),
\]}%
where the subscript of $X_{i+1}$ is taken modulo~$4$.
Similarly, we define
\tightdisplay{\[
  \psi_V(x,y)
  \;\coloneq\;
  \bigvee\nolimits_{i,j\in[0,3]}
  \Bigl(
    X_i(x)
    \wedge X_i(y)
    \wedge Y_j(x)
    \wedge Y_{j+1}(y)
    \wedge E_j(x,y)
    \wedge \neg E_{j+1}(x,y)
  \Bigr),
\]}%
where the subscript of $Y_{j+1}$ is again taken modulo~$4$.

Finally, we redefine the auxiliary formula identifying copies of the
same grid point by
\tightdisplay{\[
  \psi_E(x,y)
  \;\coloneq\;
  E_1(x,y)\wedge E_2(x,y).
\]}%
Since $E_1$ and $E_2$ are equivalence relations, $\psi_E$ defines an
equivalence relation as well.

We now use these formulas in place of the corresponding formulas from
the $\MON{}+\TRrefl$ construction, leaving all the remaining
constraints unchanged.
The same argument as before applies:
the conditions in~\eqref{eqn:infinite-tiling-unique-H} and~\eqref{eqn:infinite-tiling-unique-V} ensure
that the squares supplied by~\eqref{eqn:infinite-tiling-square} can be
joined consistently up to $\psi_E$-equivalence, while
\eqref{eqn:infinite-tiling-agree-on-all} ensures that
$\psi_E$-equivalent elements agree on their coordinate and colour
predicates.
Consequently, the resulting structure represents an
$\N\times\N$ grid carrying a valid tiling exactly when $\cT$ is
solvable.

\end{document}